\documentclass{article}
\usepackage{amsmath}
\usepackage{amsthm}
\usepackage{xcolor}
\usepackage[utf8]{inputenc}
\usepackage[ruled]{algorithm}
\usepackage{algpseudocode}
\usepackage{pgfplots}

\pgfplotsset{compat=1.17}
\newcommand{\bigO}{\mathcal{O}}
\newcommand{\sign}{{\rm sign}}

\DeclareMathOperator*{\argmax}{arg\,max}

\DeclareMathOperator{\ecc}{ecc}

\theoremstyle{definition}
\newtheorem{definition}{Definition}[section]
\theoremstyle{plain}
\newtheorem{theorem}{Theorem}[section]
\newtheorem{lemma}[theorem]{Lemma}
\newtheorem{proposition}[theorem]{Proposition}

\title{Computable Bounds and Monte Carlo Estimates of the Expected Edit Distance\footnote{ A preliminary version appeared in Nieves R. Brisaboa and Simon
     J. Puglisi, editors, String Processing and Information Retrieval,
    pages 91–106, Cham, 2019. Springer International Publishing.}}
\author{Gianfranco Bilardi\\
\normalsize{Department of Information Engineering, University of Padova, Italy}\\
\small{\texttt{bilardi@dei.unipd.it}}
\and
Michele Schimd\\
\normalsize{Department of Information Engineering, University of Padova, Italy}\\
\small{\texttt{schimdmi@dei.unipd.it}}}

\begin{document}
\maketitle

\begin{abstract}
The edit distance is a metric of dissimilarity between strings, widely
applied in computational biology, speech recognition, and machine
learning. Let $e_k(n)$ denote the average edit distance between
random, independent strings of $n$ characters from an alphabet of size
$k$.  For $k \geq 2$, it is an open problem how to efficiently
compute the exact value of $\alpha_{k}(n) = e_k(n)/n$ as well as of
$\alpha_{k} = \lim_{n \to \infty} \alpha_{k}(n)$, a limit known to
exist.

This paper shows that $\alpha_k(n)-Q(n) \leq \alpha_k \leq
\alpha_k(n)$, for a specific $Q(n)=\Theta(\sqrt{\log n / n})$, a
result which implies that $\alpha_k$ is computable. The exact
computation of $\alpha_k(n)$ is explored, leading to an algorithm
running in time $T=\bigO(n^2k\min(3^n,k^n))$, a complexity that
makes it of limited practical use.

An analysis of Monte Carlo estimates is proposed, based on McDiarmid's
inequality, showing how $\alpha_k(n)$ can be evaluated with good
accuracy, high confidence level, and reasonable computation time, for
values of $n$ say up to a quarter million. Correspondingly, 99.9\%
confidence intervals of width approximately $10^{-2}$ are obtained for
$\alpha_k$.

Combinatorial arguments on edit scripts are exploited to analytically
characterize an efficiently computable lower bound $\beta_k^*$ to
$\alpha_k$, such that $ \lim_{k \to \infty} \beta_k^*=1$. In general,
$\beta_k^* \leq \alpha_k \leq 1-1/k$; for $k$ greater than a few
dozens, computing $\beta_k^*$ is much faster than generating good
statistical estimates with confidence intervals of width
$1-1/k-\beta_k^*$.

The techniques developed in the paper yield improvements on most
previously published numerical values as well as results for alphabet
sizes and string lengths not reported before.
\end{abstract}



\section{Introduction}
Measuring dissimilarity between strings is a fundamental problem in
computer science, with applications in computational biology, speech
recognition, machine learning, and other fields. One commonly used
metric is the \emph{edit distance} (or \emph{Levenshtein distance}),
defined as the minimum number of substitutions, deletions, and
insertions necessary to transform one string into the other.

It is natural to ask what is the expected distance between two randomly
generated strings, as the string size grows; knowledge of the
asymptotic behavior has proved useful in computational biology
(\cite{GMR16}) and in nearest neighbor search
(\cite{Rub18}), to mention a few examples. 

In computational biology, the question often arises whether two
strings (\emph{e.g.}, two DNA reads) are noisy copies of the same
source or of non-overlapping sources. In several cases of interest,
the source is modeled as a sequence of independent and identically
distributed symbols (see, \emph{e.g.}, \cite{GMR16},
\cite{CS14}, with reference to DNA) and the noise is modeled with
substitutions, insertions, and deletions (a good approximation for
technologies like PacBio and MinION
\cite{WCW+17}). Then, the statistical inference may be based on a
comparison of the distance between the observed strings with either
the expected distance between a string and a noisy copy of itself, or
the expected distance between two random strings.

Even for the case of uniform and independent strings, the study of the
expected edit distance appears to be challenging and little work has
been reported on the problem.  In contrast, the closely related
problem of computing the expected length of the \emph{longest common
subsequence} has been extensively studied, since the seminal work by
\cite{CS75}.

Using Fekete's lemma, it can be shown that both metrics tend to grow
linearly with the string size $n$ (\cite{Ste97}). Specifically,
let $e_k(n)$ denote the expected edit distance between two random,
independent strings of length $n$ on a $k$-ary alphabet; then
$\alpha_k(n)=e_k(n)/n$ approaches (from above) a limit $\alpha_k \in
[0,1]$.  Similarly, let $\ell_k(n)$ denote the expected length of the
longest common subsequence; then $\gamma_k(n)=\ell_k(n)/n$ approaches
(from below) a limit $\gamma_k \in [0,1]$. The $\gamma_k$'s are known
as the Chv\'atal-Sankoff constants.  The efficient computation of
the exact values of $\alpha_k$ and $\gamma_k$ is an open problem. This
paper establishes the computability of $\alpha_k$, for any $k$,
and proposes methods for estimating and bounding $\alpha_k$,
also reporting numerical results for various alphabet sizes $k$.

From the perspective of computational complexity, we remark that, for the 
problem of computing $\alpha_k(n)$, given input $n$ (for a fixed $k$), only algorithms that run in doubly exponential time and use exponential space are currently known, including those presented in this paper. Observe that the input size is $\log_2 n$, the number of bits needed to specify the problem input $n$. Similar statements hold for the computation of $\gamma_k(n)$. Therefore, at the state of the art, we can
place these problems in the complexity class EXPSPACE $\subseteq$ 2-EXPTIME, but not in EXPTIME and, a fortiori, not in PSPACE $\subseteq$ EXPTIME\begin{footnote}{Technically, these are traditionally defined as classes of decision problems, hence our statements strictly apply to suitable decision versions of computing the constants of interest.}\end{footnote}. Analogous considerations can be made for the problem of computing the $\nu$ most significant bits of $\alpha_k$ (or $\gamma_k$). Here, the input size is $\log_2 \nu$ and the time becomes triple exponential in the input size.
Whether these problems inherently exhibit high complexity or they can be solved efficiently by exploiting a not yet uncovered deeper structure remains to be seen.

\paragraph*{Related work}
There is limited literature directly pursuing bounds and estimates for
$\alpha_k$.  It is also interesting to review results on $\gamma_k$:
on the one hand, bounds to $\gamma_k$ give bounds to $\alpha_k$; on
the other hand, techniques for analyzing $\gamma_k$ can be adapted for
analyzing $\alpha_k$.

The only published estimates of $\alpha_k$ can be found in
\cite{GMR16} which gives $\alpha_4 \approx 0.518$ for the quaternary
alphabet and $\alpha_2 \approx 0.29$ for the binary alphabet.
Estimates of $\gamma_k$ are given by \cite{Bun01}, in
particular $\gamma_2 \approx 0.8126$ and $\gamma_4 \approx 0.6544$. A
similar value is reported by \cite{Dan94} which gives
$\gamma_2 \approx 0.8124$. Estimates of $\gamma_k$ by sampling are given by
\cite{KK13}; their conjecture that $\gamma_2>0.82$ appears
to be at odds with the estimate in \cite{Bun01}. In \cite{BC22}, the conjecture $\gamma_2 \approx 0.8122$ is proposed. They also derive a closed form for the limit constant when only one
string is random and the other is a \emph{periodic} string containing all symbols of $\Sigma_k$. 

The best published analytical lower bounds to $\alpha_k$ are $\alpha_4
\geq 0.3383$ for a quaternary alphabet and $\alpha_2 \geq 0.1578$ for
a binary alphabet \cite{GMR16}. To the best of our knowledge, no
systematic study of upper bounds to $\alpha_k$ has been published.
The best known analytical lower and upper bounds to $\gamma_2$ are
given by \cite{Lue09}, who obtained $0.7881 \leq \gamma_2 \leq
0.8263$. For larger alphabets, the best results are given by
\cite{Dan94}, including $0.5455 \leq \gamma_4 \leq 0.7082$. From known
relations between the edit distance and the length of the longest
common subsequence, it follows that $1-\gamma_k \leq \alpha_k \leq
2(1-\gamma_k)$. Thus, upper and lower bounds to $\alpha_k$ can be
respectively obtained from lower and upper bounds to $\gamma_k$.  From
$\gamma_2 \leq 0.8263$ of \cite{Lue09}, we obtain $\alpha_2 \geq 0.1737$,
which is tighter than the bound given in \cite{GMR16}. Instead
$\gamma_4 \leq 0.7082$ of \cite{Dan94} yields $\alpha_4 \geq 0.2918$,
which is weaker than the bound $\alpha_4 \geq 0.3383$ \cite{GMR16}.
From the weaker relation $(1-\gamma_2)/2 \leq \alpha_2$, \cite{Rub18}
obtained the looser bound $\alpha_2 \geq 0.0869$.  In this paper, we
derive improved bounds, for both $\alpha_2$ and $\alpha_4$, as well as
bounds on $\alpha_k$, for values of $k$ not fully addressed by earlier
literature. Some of our techniques resemble those used in \cite{BGNS99} for estimating $\gamma_k$.
Table~\ref{tab:alpha:lb:compare} shows lower bounds to $\alpha_k$ for
various values of $k$ based on this work and on that of previous
authors.  Lower bounds from \cite{Dan94} and \cite{Lue09} are obtained
from upper bounds to $\gamma_k$, which we have translated into
$\alpha_k \geq 1 - \gamma_k$. Dan\v{c}\'ik reported values of the
bound only for $k \leq 15$.  Lueker reported only the numerical upper bound
to $\gamma_2$; his computational approach is interesting and
sophisticated, but its time and space are exponential with
$k$. Experimenting with (a minor adaptation of) the software
provided by the author, we have not been able to compute $\gamma_3$
within reasonable time. We have obtained the values reported in the
Ganguly \emph{et al.}  column by numerically solving their equations
(in \cite{GMR16}, only the values for $k=2$ and $k=4$ were
reported). The equations underlying the results in the rightmost
column of Table~\ref{tab:alpha:lb:compare} are developed in
Section~\ref{sec:lowerbound}, together with a rigorous analysis of
their numerical solution.

To assess the tightness of bounds to $\gamma_k$, several authors have
investigated the \emph{rate of convergence} of $\gamma_k(n)$ to
$\gamma_k$. The bound $0 \leq\gamma_k- \gamma_k(n)\leq
\bigO(\sqrt{\log{n}/n})$ has been obtained by \cite{Ale94}
and, with a smaller constant, by
\cite{LMT12}. \cite{Lue09} introduced a sequence of upper
bounds $\gamma_k^h$ converging to $\gamma_k$ and satisfying $0 \leq
\gamma_k^h - \gamma_k \leq \bigO((\log{h}/h)^{1/3})$, where the time
complexity and the space complexity of computing
$\gamma_k^h$ increase exponentially with $h$. Observing that $h$ is in turn exponential in the number $\nu$ of desired bits for $\gamma_k$ and that $\nu$ is exponential in the input size $\lceil \log_2 \nu\rceil$, we see that computation time is a triple exponential.  No study of the
rate of convergence of $\alpha_k(n)$ to $\alpha_k$ has been
published. In this paper, we show that $0 \leq \alpha_k(n)-\alpha_k
\leq \bigO(\sqrt{\log{n}/n})$, exploiting a framework developed in
\cite{LMT12}.

Recently, \cite{Tis22} has established that
$\gamma_2$ is an algebraic number, introducing novel
ideas, which may open new perspectives on the analysis of $\gamma_k$ and $\alpha_k$, for any $k$.

\begin{table*}[t]
    \centering
    \caption{Comparison of lower bounds to $\alpha_k$ obtained in this 
    paper
    and in previous work.
    Best known bounds are highlighted in bold face.}
    \label{tab:alpha:lb:compare}
    \vspace*{3mm}
    \begin{tabular}{ccccc}
    \hline
    $k$ & Dan\v{c}\'ik & Lueker & Ganguly \emph{et al} & This work\\
    \hline
    $2$  & $0.162377$ & $\mathbf{0.17372}$ & $0.157761$ & $0.170552$ \\
    $3$  & $0.234197$ & - & $0.265028$ & $\mathbf{0.283660}$ \\
    $4$  & $0.291764$ & - & $0.338322$ & $\mathbf{0.359783}$ \\
    $5$  & $0.335572$ & - & $0.392040$ & $\mathbf{0.415173}$ \\
    $6$  & $0.370684$ & - & $0.433508$ & $\mathbf{0.457766}$ \\
    $7$  & $0.399816$ & - & $0.466732$ & $\mathbf{0.491836}$ \\
    $8$  & $0.424593$ & - & $0.494136$ & $\mathbf{0.519901}$ \\
    $16$ & - & - & $0.616273$ & $\mathbf{0.644758}$ \\
    $32$ & - & - & $0.708537$ & $\mathbf{0.738677}$ \\
    
    \hline
    \end{tabular}
\end{table*}

\paragraph*{Paper contributions and organization}
The notation and definitions used throughout this paper are given in Section \ref{sec:preliminaries}.  In
Section \ref{sec:computability}, an upper bound $\alpha_k(n)-\alpha_k
\leq Q(n)$ is derived, for each $k \geq 2$, where $Q(n)=
\Theta(\sqrt{\log{n}/n})$ is a precisely specified function
(independent of $k$). This implies $\alpha_k \in
[\alpha_k(n)-Q(n),\alpha_k(n)]$, where the interval can be made
arbitrarily small by choosing a suitably large $n$. One corollary is
the computability of the real number $\alpha_k$, for each $k \geq 2$.
Unfortunately, the algorithm underlying the computability proof is
of little practical use, since the only known method to exactly
compute $\alpha_k(n)$ is by direct application of its definition,
resulting in $\bigO(n^2k^{2n})$ time. Even after some improvement
presented in Section \ref{sec:upperbound}, the upper bound
$\alpha_{k} \leq \alpha_k(n)$ is practically computable only for
small values of $k$ and $n$. Moreover, for the feasible values of
$n$, $Q(n)$ is too large for the lower bound $\alpha_k \geq
\alpha_k(n)-Q(n)$ to be useful. These considerations motivate the
exploration of alternate approaches.

In Section \ref{sec:montecarlo}, an analysis, based on McDiarmid's
inequality, is developed for Monte Carlo estimates of $\alpha_k(n)$
obtained from the edit distance of a sample of $N$ pairs of
strings. The analysis yields the radius $\Delta$ of confidence
intervals for $\alpha_k(n)$, in terms of $n$, $N$, and the desired
confidence level $\lambda$. The (sequential) time to obtain an
estimate can be approximated as $T\approx \tau_{ed}
\frac{n}{\Delta^2}\ln^2 \left(\frac{1}{1-\lambda}\right)$, where
$\tau_{ed}$ is of the order of $5ns$, on a typical state of the art
processor core. Rather large values of $n$ can then be dealt with.  As
an indication, a $\lambda=0.999$ confidence interval of radius $\Delta=
0.67~10^{-3}$ is obtained for $\alpha_k(2^{15})$ in about $43$ minutes.
The corresponding confidence interval for $\alpha_k$ has radius
$\Delta+\frac{Q(2^{15})}{2}= 0.00068 + 0.01320 = 0.01388$.

In Section \ref{sec:upperbound}, upper bounds to $\alpha_k$ by exact
computation of $\alpha_k(n)$ for small values of $n$ are obtained, by
introducing an $\bigO(n^2(3k)^{n})$ time algorithm that, while still
exponential in $n$, is (asymptotically and practically) faster than
the straightforward, $\bigO(n^2(k^2)^{n})$ time, algorithm. When $k$
is of the order of a few dozens, only very small values of $n$ are
feasible and $\alpha_k(n)$ does not differ appreciably from the
quantity $1-\frac{1}{k}$, which satisfies $\alpha_{k}(n) \leq
1-\frac{1}{k}$, as it can be easily shown by allowing only
substitutions (cf. Hamming distance).

In Section \ref{sec:lowerbound}, a lower bound $\alpha_k \geq
\beta_k^{*}$ is established, for each $k \geq 2$.  A counting argument
provides a lower bound to the number of string pairs with distance at
least $\beta n$; an asymptotic analysis provides conditions on $\beta$
under which the contribution to $\alpha_k(n)$ of the remaining string
pairs vanishes with $n$. A careful study leads to a numerical
algorithm to compute $\beta_k^{*}$, the supremum of the $\beta$'s
satisfying such conditions, with any desired accuracy, $\epsilon$.
Since, as shown in Section \ref{sec:lowerbound}, $\lim_{k \rightarrow
  \infty} \beta_k^{*} = 1$, the interval
$[\beta_k^{*},1-\frac{1}{k}]$, which contains $\alpha_k$, has size
vanishing with increasing $k$.  For $k$ large enough, it becomes a
subset of a confidence interval obtained with comparable computational
effort.  As an example, $\beta_{2^{40}}^{*} \approx 0.999984 \leq
\alpha_{2^{40}} \leq 1-2^{-40} \approx 0.999999$, placing
$\alpha_{2^{40}}$ in an interval of size smaller than
$0.16~10^{-4}$. To achieve $Q(n) \leq 0.16~10^{-4}$ requires $n \geq
10^{11}$.  On a single core, computing the edit distance for just
one pair of strings of length $n=10^{11}$ would take time $T \approx
5~10^{-9}10^{22}s \approx 1.6~10^6~years$, whereas computing
$\beta_{2^{40}}^{*}$ took just 14 milliseconds, using a straightforward, non-optimized
implementation.

By applying the above methodologies, we numerically derive guaranteed
as well as statistical estimates for specific $\alpha_k$'s and
$\alpha_k(n)$'s. In particular, Table~\ref{tab:reults:summary}
summarizes our numerical results for various alphabet sizes.  For each
$k$, the table reports an interval that provably contains $\alpha_k$
and a (narrower) interval that contains $\alpha_k$ with confidence
$0.999$.  For the guaranteed interval, more details are provided in
Table~\ref{tab:lower:bounds:k}, Section~\ref{sec:lowerbound} (left
endpoint) and in Table~\ref{tab:upper:bounds},
Section~\ref{sec:upperbound} (right endpoint).  For the confidence
interval, see also Table~\ref{tab:hat:alpha:k},
Section~\ref{sec:montecarlo}.

In Section~\ref{sec:large:k:conjecture}, we wonder about the
asymptotic behavior of $\alpha_k$, with respect to $k$. We propose
and motivate the conjecture that $\lim_{k \rightarrow \infty}
(1-\alpha_k)k= c_{\alpha}$ for some constant $c_{\alpha} \geq 1$.
Numerical evidence indicates that perhaps $3 \leq c_{\alpha} \leq
4$.
  
Finally, Section \ref{sec:conclusion} presents conclusions and further
directions of investigation.

\begin{table*}[t]
    \centering
    \caption{Summary of numerical results obtained applying the methodologies
    presented in this paper. For various sizes $k$, the table shows
    an interval guaranteed to contain
    $\alpha_k$ and an interval containing
    $\alpha_k$ with confidence $0.999$.}
    \label{tab:reults:summary}
    \vspace*{3mm}
    \begin{tabular}{ccc}
    \hline
    $k$  & Guaranteed interval & Confidence interval\\
    \hline
    \\[-3.5mm]
    $2$  & $[0.17055, 0.36932]$ & $[0.26108, 0.28884]$\\
    $3$  & $[0.28366, 0.53426]$ & $[0.40144, 0.42920]$\\
    $4$  & $[0.35978, 0.63182]$ & $[0.49031, 0.51807]$\\
    $5$  & $[0.41517, 0.70197]$ & $[0.55289, 0.58066]$\\
    $6$  & $[0.45776, 0.75149]$ & $[0.60002, 0.62778]$\\
    $7$  & $[0.49183, 0.79031]$ & $[0.63701, 0.66477]$\\
    $8$  & $[0.51990, 0.81166]$ & $[0.66694, 0.69470]$\\
    $16$ & $[0.64475, 0.89554]$ & $[0.79198, 0.81974]$\\
    $32$ & $[0.73867, 0.96588]$ & $[0.87230, 0.90007]$\\
    \hline
    \end{tabular}
\end{table*}

This paper expands over the conference version \cite{SB19}; additions
include: (i) an analysis of the rate of convergence of $\alpha_k(n)$
to $\alpha_k$; (ii) a novel confidence-interval analysis for the
Monte Carlo estimate of $\alpha_k(n)$; (iii) a rigorous
development and a proof of correctness for an algorithm that can
numerically compute the lower bound $\beta_k^*$ to $\alpha_k$, with any
desired accuracy; and (iv) a conjecture on the behavior of $\alpha_k$ for large $k$.

\section{Preliminaries}
\label{sec:preliminaries}
In this section, we introduce the notation adopted throughout the paper
and present some preliminary definitions and results used in various
parts of the work.
\subsection{Notation and definitions}\label{sec:notationsAndDef}
Let $\Sigma_k$ be a finite alphabet of size $k\geq 2$ and let $n\geq
1$ be an integer; a \emph{string} $x$ is a sequence of symbols $x[1]
x[2] \ldots x[n]$ where $x[i] \in \Sigma_k$; $n$ is called the
\emph{length} (or \emph{size}) of $x$, also denoted by $|x|$. 
$\Sigma_k^n$ is the set of all strings of length $n$.

\paragraph*{Edit distance}
We consider the following \emph{edit operations} on a string $x$: the
\emph{match} of $x[i]$, the \emph{substitution} of
$x[i]$ with a different symbol $b\in \Sigma_k\setminus\{x[i]\}$, the
\emph{deletion} of $x[i]$, and the \emph{insertion} of $b \in\Sigma_k$
in position $j=0,\ldots,n$ (insertion in $j$ means $b$ goes after
$x[j]$ or at the beginning if $j=0$); an \emph{edit script} is a
sequence of edit operations. With each type of edit operation is
associated a cost; throughout this paper, matches have cost $0$ and
other operations have cost $1$.  The cost of a script is the sum of
the costs of its operations.  The \emph{edit distance} between $x$ and
$y$, $d_E(x,y)$, is the minimum cost of any script transforming $x$
into $y$.  It is easy to see that $||x|-|y||\leq d_E(x,y)\leq
\max(|x|,|y|)$.

\paragraph*{Simple scripts}
We can view a string as a sequence of \emph{cells}, each
containing a symbol from $\Sigma_k$, and consider edit operations as
acting on such cells: a deletion destroys a cell, a substitution
changes the content of a cell, and an insertion creates a new cell
with some content in it (matches leave cells untouched). We will say that a script is \emph{simple}
if it performs at most one edit operation on each cell. It is easy
to see that, if a script transforming $x$ into $y$ is not simple,
then there is a script with fewer operations which achieves the same
transformation. In fact, if a cell is eventually deleted, any
operation performed on it prior to its deletion can be safely
removed from the script; if a cell is inserted, any subsequent
substitution can be removed, appropriately selecting the content of
the initial insertion; and multiple substitutions on a cell that is
retained can be either replaced by just one appropriate substitution
or removed altogether. Thus, a script of minimum cost is necessarily
simple, so that, to determine $d_E(x,y)$, we can restrict our
attention to simple scripts.

\paragraph*{Scripts and alignments}
Given an edit script transforming $x$ into $y$, consider those
cells of $x$ that are retained in $y$, possibly with a different
content. Since the relative order of two such cells is the same in
$x$ and in $y$, the positions occupied by such cells in $x$ and $y$
form an alignment, in the sense defined next.

An \emph{alignment} $(\mathcal{I},\mathcal{J})$ between $x$ and $y$ is
a pair of increasing integer sequences of the same length $s$
\begin{align*}
\mathcal{I} &= (i_1, \ldots, i_s) \qquad
    1 \leq i_1 < i_2 < \ldots < i_s \leq |x|,\\
\mathcal{J} &= (j_1,\ldots,j_s) \qquad
    1 \leq j_1 < j_2 < \ldots < j_s \leq |y|.
\end{align*}
The positions $i_\ell$ in $x$ and $j_\ell$ in $y$ are said to be to be
aligned in $(\mathcal{I},\mathcal{J})$. To each script
$\mathcal{S}$, there corresponds a unique alignment
$a(\mathcal{S})=(\mathcal{I},\mathcal{J})$, where $s$ equals the
number of cells of $x$ that are retained in $y$ and, for every
$\ell=1,2, \ldots, s$, the cell in position $i_{\ell}$ of $x$ has
moved to position $j_{\ell}$ of $y$. If $\mathcal{S}$ is simple,
then: (i) for aligned positions $i_\ell$ and $j_\ell$, $y[j_\ell]$
is substituted with or matched to $x[i_\ell]$ depending on whether
$y[j_\ell] \neq x[i_\ell]$ or not; (ii) for positions $i \notin
\mathcal{I}$, $x[i]$ is deleted, and (iii) for positions $j \notin
\mathcal{J}$, $y[j]$ is inserted.

Next, we prove a simple lemma, which will be useful both in
Section~\ref{sec:computability}, to cast edit distance within the
framework of~\cite{LMT12}, and in Section~\ref{sec:lowerbound}, to
develop a counting argument leading to a lower bound on $\alpha_k$.
\begin{lemma}
\label{lemma:simpleScript}
With the preceding notation, if $\mathcal{S}$ is a simple script to
transform $x$ into $y$, with $|x|=|y|=n$, and
$(\mathcal{I},\mathcal{J})=a(\mathcal{S})$ is the corresponding
alignment, its cost is
\begin{equation}\label{eq:simpleCost}
cost(x,y,\mathcal{S})=2(n-s)+\sum_{\ell=1}^{s}(x[i_\ell] \not= y[j_\ell]).
  \end{equation}
\end{lemma}
\begin{proof}
The script performs $(n-s)$ deletions, $(n-s)$ insertions, and
$\sum_{\ell=1}^{s}(x[i_\ell] \not= y[j_\ell])$ substitutions.
\end{proof}
We may observe that (for given $x$ and $y$) simple scripts
corresponding to the same alignment differ only with respect to the
order in which the edit operations are applied. Such an order does not affect
the final result, since in a simple script different operations act on different cells.
Thus, the number of distinct simple scripts with the same alignment is
the factorial of their cost (\emph{i.e.}, of the number of edit
operations), given by Equation~(\ref{eq:simpleCost}).

\paragraph*{Random strings and the limit constant}
A \emph{random string} of length $n$, $X = X[1] X[2]
\ldots X[n]$, is a sequence of \emph{random symbols} $X[i]$ generated
according to some distribution over $\Sigma_k$. We will assume that
the $X[i]$'s are uniformly and independently sampled from $\Sigma_k$
or, equivalently, that $\Pr{[X=x]} = k^{-n}$ for every
$x\in\Sigma_k^n$. We define the \emph{eccentricity} $\ecc(x)$ of a
string $x$ as its expected distance from a random string
$Y\in\Sigma_k^n$:

\begin{align} \label{eq:eccentricity}
\ecc(x) = k^{-n}\sum_{y\in\Sigma_k^n}{d_E(x,y)}.
\end{align}
The expected edit distance between two random, independent
strings of $\Sigma_k^n$ is:
\begin{align} \label{eq:ed:eccentricty}
e_k(n)  &= k^{-2n}\sum_{x\in\Sigma_k^n}{\sum_{y\in\Sigma_k^n}{d_E(x,y)}} \nonumber \\
        &= k^{-n}\sum_{x\in\Sigma_k^n}{\ecc(x)}.
\end{align}
Let $\alpha_k(n)=e_k(n)/n$; it can be shown (Fekete's lemma from
ergodic theory; see, \emph{e.g.}, Lemma 1.2.1 in \cite{Ste97}) that
there exists a real number $\alpha_k \in [0,1]$, such that
\begin{align}
\label{eq:alpha:definition}
\lim_{n\rightarrow\infty}{\alpha_k(n)} = \alpha_k .
\end{align}
The main objective of this paper is to derive estimates and bounds to
$\alpha_k$.

\paragraph*{Rate of convergence to the limit constant}  In the
  outlined context, it is of interest to develop upper bounds, as
  functions of $n$, to the quantity
\begin{align*}
  q_k(n) = \alpha_k(n) - \alpha_k,
\end{align*}
which we will refer to as the \emph{rate of convergence}, following a
terminology widely used for analogous quantities in the context of the
longest common subsequence (\emph{e.g.}, \cite{Ale94, LMT12}).

\subsection{Computing the edit distance}
\label{sec:DPalgorithms}
The edit distance and the length of the longest common subsequence
(LCS) can be computed by a \emph{dynamic programming} algorithm.
Given two strings, $x$ of length $n$ and $y$ of length $m$, their edit
distance $d_E(x,y)$ is obtained as the entry $M_{n,m}$ of an
$(n+1)\times(m+1)$ matrix $\mathbf{M}$, computed according to the
following recurrence:
\begin{align} \label{eq:edit:distance:recurrence}
&M_{i,0}=i &\textrm{for} \ i=0,\ldots,n\nonumber \\
&M_{0,j}=j &\textrm{for} \ j=0,\ldots,m \\
&M_{i,j} = \min{( M_{i-1,j-1} + \xi_{i,j} , M_{i-1,j}+1 , M_{i,j-1}+1 )}  
    &\textrm{for} \ {i>0} \ \textrm{and} \ {j>0}\nonumber
\end{align}
where $\xi_{i,j}=0$ if $x[i] = y[j]$ and $\xi_{i,j}=1$
otherwise.\footnote{A similar algorithm computes the length of the
LCS. Recurrence (\ref{eq:edit:distance:recurrence}) becomes
$M_{i,0} = 0$, $M_{0,j} = 0$, and $M_{i,j} = \max{( M_{i-1,j-1} +
  (1-\xi_{i,j}) , M_{i-1,j} , M_{i,j-1} )}.$} This algorithm takes
$\bigO(nm)$ time and space.  An edit script transforming $x$ into $y$
can be obtained \emph{backtracking} on $\mathbf{M}$, along a
\emph{path} from cell $(n,m)$ to cell $(0,0)$. For both edit distance
and LCS, the approach by \cite{MP80}, exploiting
the method of the Four Russians, reduces the time to
$\bigO(\frac{n^2}{\log{n}})$, assuming $n\geq m$.  Although
asymptotically faster, the algorithm in \cite{MP80} is seldomly used.
Other approaches are
usually preferred, such as the one proposed by \cite{Mye99},
which reduces the $n^2$ bound by a factor proportional to the machine word size, implementing Recurrence (\ref{eq:edit:distance:recurrence})
via bit-wise operations.

The space complexity of the basic dynamic programming algorithm can be
reduced to $\bigO(\min(n, m))$.  Assuming, w.l.o.g., that $n \geq m$,
simply proceed row-wise storing only the last complete row. While this
approach is not directly amenable to constructing scripts by
backtracking, a more sophisticated divide and conquer version due to
\cite{Hir75} yields the edit distance and an edit script in
quadratic time and linear space. \cite{KR08}
applied the Four Russians method to Hirschberg's algorithm, improving
its running time by a logarithmic factor.

It is known that both the edit distance and the length of the LCS
cannot be computed in time $\bigO(n^{2-\epsilon})$, unless the
\emph{Strong Exponential Time Hypothesis (SETH)} is false (\cite{ABW15}, \cite{BI15}).

Approximate computation of the edit distance has been extensively
studied. \cite{Ukk85} presents a banded algorithm that 
computes an approximation within a factor
$\bigO(n^{1-\epsilon})$ in time
$\bigO(n^{1+\epsilon})$. Interestingly, this algorithm computes the
exact distance whenever such distance is $\bigO(n^\epsilon)$
(although, it may output the exact distance also for higher
values). \cite{LMS98} give an algorithm that
computes the exact distance in time $\bigO(n+d^2)$, where $d$ is the
distance itself. Thus, sub-quadratic time can be achieved when the
distance is sub-linear.  More recently, an
$(\log{n})^{\bigO(1/\epsilon)}$ approximation, computable in time
$\bigO(n^{1+\epsilon})$, was proposed by
\cite{AKO10}, and a constant approximation algorithm with running time
$\bigO(n^{1+5/7})$ was proposed by
\cite{CDG+18, CDG+20}.  The work by \cite{RS19}
gives a reduction from approximate length of the longest common
subsequence to approximate edit distance, proving that the algorithm
in \cite{CDG+18} can also be used to approximate the length of the
LCS.

In order to compute upper bounds to $\alpha_k$, we propose an
algorithm related to the approaches developed by \cite{COH17} and \cite{Lue09}.  In these works,
portions of the dynamic programming matrix are associated to the states of
a finite state machine. Our algorithm conceptually simulates all
possible executions of a machine similar to the one defined in
\cite{COH17}.

\section{Rate of convergence and computability of $\alpha_k$}
\label{sec:computability}
For each $n$, $\alpha_k(n)$ is a rational number, which can be
computed, according to Equation~(\ref{eq:ed:eccentricty}), by
exhaustively enumerating all pairs $(x,y)$ of strings in $\Sigma_k^n$
and accumulating the corresponding values $d_E(x,y)$, which can be
obtained with any algorithm for the exact edit distance.  On the other
hand, the limit constant $\alpha_k$ is known to exist as a real
number, whose rationality remains an open problem.  In this section,
we show that, for every $k$, this number is computable, according to
the following (standard) definition.
\begin{definition}[Computability of a real number]
\label{def:computability}
A real number $\zeta$ is \emph{computable} if there exists an
algorithm that, given as input a rational number $\epsilon > 0$,
outputs a rational number $\zeta_{\epsilon}$, such that
$|\zeta-\zeta_{\epsilon}| < \epsilon$.
\end{definition}
As shown by \cite{BBD21}, the subadditivity of a
sequence of rational numbers, while sufficient to guarantee the
existence of a limit (Fekete's Lemma, \cite{Spe14}), is not sufficient to guarantee its computability
which, if present, requires an additional argument.  For $\alpha_k$,
such an argument can be provided along the following steps:
\begin{itemize}
\item Prove that, for some computable function $b_k(n)$, we have
\begin{align}
\label{eq:alpha:convergenceUB}
\alpha_k(n) - \alpha_k \leq b_k(n) \qquad \forall n \geq 1.
\end{align}
\item Show that there is an algorithm which, given a rational number
  $\epsilon>0$, computes an integer $n_{\epsilon} \geq 1$ such that
\begin{align}
\label{eq:alpha:epsilonUBB}
b_k(n) < \epsilon \qquad \forall n \geq n_{\epsilon}.
\end{align}
\item Let $\alpha_{k,\epsilon}=\alpha_k(n_{\epsilon})$ and observe
  that
\begin{align}
  \label{eq:alpha:epsilon}
  0 < \alpha_{k,\epsilon}-\alpha_k < b_k(n_{\epsilon})<\epsilon,
\end{align}
thus complying with Definition~\ref{def:computability} (with
$\zeta=\alpha_k$ and $\zeta_{\epsilon}=\alpha_{k,\epsilon}$), since
$n_{\epsilon}$ is computable from $\epsilon$ and
$\alpha_{k,\epsilon}=\alpha_k(n_{\epsilon})$ is computable from
$n_{\epsilon}$.
\end{itemize}

To obtain a bound $b_k(n)$ such that
Equation~(\ref{eq:alpha:convergenceUB}) is satisfied, we show how the
edit distance problem can be cast within a framework developed in
\cite{LMT12}, to analyze the limit average behavior of certain
functions of random string pairs on a finite alphabet $\Sigma$.  These
functions are formally defined next.
\begin{definition}
\label{def:scoring:scheme}
Let $S:\Sigma \times \Sigma \rightarrow R_{0}^{+}$ be a symmetric
($S(b,a)=S(a,b)$), non-negative, real function and let $\delta$ be a
real number. The \emph{score} of a pair of strings $x,y\in\Sigma^n$,
with respect to a given alignment $(\mathcal{I}, \mathcal{J})$, is
defined as
\begin{align*}
S(x,y;\mathcal{I}, \mathcal{J}) = 
    \sum_{\ell=1}^{s}{S(x[i_\ell],y[j_\ell])} + \delta(n-s).
\end{align*}
The score of the string pair is the maximum score over all possible
alignments:
\begin{align*}
  S(x,y) = \max_{(\mathcal{I},\mathcal{J})} S(x,y;\mathcal{I}, \mathcal{J}).
\end{align*}
\end{definition}
Next, we express the edit distance in terms of a suitable score.
\begin{proposition}\label{prop:edit-score}
For $a,b \in \Sigma_k$ let
\begin{align*}
S(a,b) = \left\lbrace\begin{array}{ll}
0 & \textrm{if} \ a \neq b\\
1 & \textrm{otherwise}\\
\end{array} \right. 
\end{align*}
and let $\delta=-1$. Then, for $x,y \in \Sigma_k^n$,
\begin{align}\label{eq:edit-cost:score}
d_E(x,y)=n - S(x,y).
\end{align}
\end{proposition}
\begin{proof}
Considering a given alignment $(\mathcal{I},\mathcal{J})$, we can
write
\begin{align*}
  n-S(x,y;\mathcal{I},\mathcal{J})
  &=n-\sum_{\ell=1}^{s}{S(x[i_\ell],y[j_\ell])} -(-1)(n-s)\\
  &=2(n-s)+\left(s-\sum_{\ell=1}^{s}{S(x[i_\ell],y[j_\ell])}\right)\\
  &=2(n-s)+\sum_{\ell=1}^{s}(x[i_\ell] \not= y[j_\ell]).
\end{align*}
Recalling Lemma~\ref{lemma:simpleScript}, the quantity in the last row
can be recognized as the cost of the simple scripts $\mathcal{S}$
transforming $x$ into $y$, with alignment
$a(\mathcal{S})=(\mathcal{I},\mathcal{J})$. Further recalling that
optimal scripts are simple, we conclude that maximizing
the score with respect to the alignment minimizes the cost of the edit
script, whence the claimed Equality~(\ref{eq:edit-cost:score}).
\end{proof}

The preceding proposition enables the application of the following far
reaching result to the analysis of the average edit distance.
\begin{theorem}[\cite{LMT12}]\label{th:LMT12}
Let $X$ and $Y$ be random strings in $\Sigma^n$, whose symbols are all
mutually independent and equally distributed. Let $l(n) =
\frac{1}{n}E_{X,Y}S(X,Y)$, where $E$ denotes the expectation operator,
and let $l=\lim_{n \to \infty}l(n)$. Then
\begin{align}
\label{eq:rate:convergence:Q}
 l-l(n) \leq A \sqrt{\frac{2}{n-1}\left(\frac{n+1}{n-1}
    + \ln(n-1)\right)} + \frac{F}{n-1} := Q_{A,F}(n).
\end{align}
where $A=\max_{a,b{\in}\Sigma}S(a,b)$ and $F=
\max_{a,b,c{\in}\Sigma}|S(a,b)-S(a,c)|$.
\end{theorem}

The preceding theorem makes no assumption on the probability
distribution of the symbols. Moreover, $Q_{A,F}(n)$ is independent of
such distribution, although the quantities $l(n)$ and $l$ are not.
Here, we assume the uniform distribution, upon which we have based the
definition of $\alpha_k(n)$.  As a corollary of Theorem~\ref{th:LMT12}
and Proposition~\ref{prop:edit-score}, we obtain the computability of
$\alpha_k$.
\begin{theorem}
\label{thm:computability:alpha}
For any integer $k>0$, the limit constant $\alpha_k$ defined in
Equation~(\ref{eq:alpha:definition}) satisfies the bound
\begin{align}\label{eq:rate:alpha:Q}
 \alpha_k(n)-\alpha_k \leq \sqrt{\frac{2}{n-1}\left(\frac{n+1}{n-1}
    + \ln(n-1)\right)} + \frac{1}{n-1} := Q(n),
\end{align}
for $n\geq2$. Therefore, $\alpha_k$ is a computable real number.
\end{theorem}
\begin{proof}
It is an exercise to see that, for the score function $S(a,b)$ of
Proposition~\ref{prop:edit-score}, $A=F=1$, for any
$k$. Correspondingly, we have $Q(n)=Q_{1,1}(n)$.  We also observe
that, with $\Sigma=\Sigma_k$, using
  Equation~(\ref{eq:edit-cost:score}), we have
\begin{align*}
\alpha_k(n)
    &=\frac{1}{n}E_{X,Y}d_E(X,Y)=1-\frac{1}{n}E_{X,Y}S(X,Y)
    =1-l_k(n)\\
\alpha_k &=1-l_k.
\end{align*}
Thus, Equation~(\ref{eq:rate:alpha:Q}) follows from
Equation~(\ref{eq:rate:convergence:Q}), with $Q(n)=Q_{1,1}(n)$,
considering that $\alpha_k(n)-\alpha_k=(1-l_k(n))-(1-l_k)=l_k-l_k(n)$.
Finally, it is straightforward to prove that there is an integer
$\bar{n}$ such that $Q(n)$ is strictly decreasing for
$n\geq\bar{n}$ and that
$n_{\epsilon}=\min\{n\geq\bar{n}:Q(n)<\epsilon\}$ is a computable
function of the rational number $\epsilon$.
\end{proof}

Interestingly, $Q(n)$, hence $n_{\epsilon}$, is
independent both of the alphabet size, $k$, and of the probabilities of symbols in the alphabet $\Sigma$. However, computing
$\alpha_k(n_{\epsilon})$, to obtain a deterministically guaranteed
$\epsilon$-approximation of $\alpha_k$, will require work increasing
with $k$, at least with the currently known approaches that are
sensitive to the number $k^n$ of strings of length $n$, as we will see
in the coming sections. Moreover, $n_{\epsilon}$
increases more than quadratically with $1/\epsilon$ (see
Equation~(\ref{eq:rate:alpha:Q})), making the approach completely
impractical. More specifically, to obtain $\nu$ bits of $\alpha_k$ we need $\epsilon = 2^{-\nu}$, so that $n_{\epsilon} \geq 1/\epsilon^2 = 4^\nu$ and $k^{n_\epsilon} \geq k^{4^{\nu}}$, which is a triple exponential in the problem size $\log_2{\nu}$.

We point out that, in the same spirit of this section, the computability of $\gamma_k$ can be derived from any of the rate-of-convergence bounds given by \cite{Ale94}, \cite{Lue09}, and \cite{LMT12}.

\section{Monte Carlo estimates of $\alpha_k$}
\label{sec:montecarlo}
In this section, motivated by the difficulty of the exact computation,
we develop an analysis of Monte Carlo estimates of $\alpha_k(n)$, by
sampling, and translate them into estimates of $\alpha_k$, using
Theorem~\ref{thm:computability:alpha}.  We will see how $\alpha_k(n)$
can be estimated with high confidence and good accuracy for values of
$n$ up to a quarter million, with less than one core-hour of
computation.  For $\alpha_k$, we achieve an error of the order of
$10^{-2}$.
  
Intuitively, for fixed $k$ and $n$, we expect the estimate error to be
proportional to $1/\sqrt{N}$, where $N$ is the number of samples (string
pairs), and to the standard deviation $S_k(n)/n$ of the single sample
$d_E(x,y)/n$. Experimentally, this standard deviation appears to
decrease a bit faster than $1/\sqrt{n}$,\footnote{We do not report
here on this experimental observation systematically; however,
estimates of $S_4(n)$ for some values of $n$ can be found in
Table~\ref{tab:mc}.} with the implication that, to obtain the same
error on $\alpha_k(n)$, fewer samples suffice for larger $n$. However,
the behavior of the standard deviation does not appear easy to
establish analytically.  Fortunately, the edit distance function has
the property that, if only one position is modified in just one of the
input strings, the (absolute value of the) variation of the distance
is at most $1$. This property enables the use of McDiarmid's
inequality to bound from above the probability that $d_E(x,y)/n$
deviates from the mean by more than a certain amount, by a negative
exponential in the square of that amount. This paves the way to the
desired analysis. In fact, McDiarmid's inequality can be applied
directly to the average over $N$ samples, dealing in a uniform way
with the ``averaging'' effect within a single pair of strings and
across multiple pairs.

\subsection{$\alpha_k(n)$}
\label{sec:alpha:k:n}
Given $N$ random and independent pairs of strings
$(X_1,Y_1), \ldots, (X_N,Y_N)$ from $\Sigma_k^n$, we consider the random variable

\begin{align} \label{eq:sample:mean}
\tilde{\alpha}_{k}(n,N)=\frac{1}{n}\tilde{e}_k(n,N) = \frac{1}{nN}{\sum_{i=1}^{N}d_E(X_i,Y_i)}.
\end{align}

Clearly, $E[\tilde{\alpha}_k(n,N)]=\alpha_k(n)$.  To assess the
quality of $\tilde{\alpha}_k(n,N)$ as an estimate for $\alpha_k(n)$, we
resort to McDiarmid's inequality, briefly reviewed next.
\begin{proposition}[\cite{McD89}]\label{prop:McD89}
Let ${\bf Z}=(Z_1,\ldots,Z_{2m})$ be a vector of $2m$ \emph{independent}
random variables. Let $f({\bf Z})$ be a real function and let $B>0$ be
a real constant such that, if ${\bf Z}$ and ${\bf Z'}$ differ in at
most one component, then $|f({\bf Z})-f({\bf Z'})| \leq B$. Then, for
every $\Delta \geq 0$,
\begin{align}\label{eq:McDiarmid}
\Pr[|f({\bf Z})-E[f({\bf Z})]| > \Delta] \leq
  \exp{\left(-\frac{\Delta^2}{mB^2}\right)}.
\end{align}
\end{proposition}
In the present context, based on the previous proposition, we can
formulate confidence intervals for $\alpha_k(n)$, as follows.
\begin{proposition}[Confidence intervals for $\alpha_k(n)$]
\label{prop:aplha:k:n:confidence}
For any $\Delta\geq 0$, the (random) interval
$[\tilde{\alpha}_{k}(n,N)-\Delta, \tilde{\alpha}_{k}(n,N)+\Delta]$
is a \emph{confidence interval} for the \emph{parameter}
$\alpha_{k}(n)$, with \emph{confidence level}
$1-2\exp{(-Nn\Delta^2)}$, that is:
\begin{align}\label{eq:alpha-McDiarmid}
  \Pr[\tilde{\alpha}_k(n,N)-\Delta \leq \alpha_k(n) \leq
    \tilde{\alpha}_k(n,N)+\Delta] \geq
  1-2\exp{\left(-Nn\Delta^2\right)}.
\end{align}
\end{proposition}
\begin{proof}
We apply Proposition~\ref{prop:McD89} to the function $f$ given by the
rightmost term in Equation~(\ref{eq:sample:mean}), with ${\bf Z}$
being the concatenation the $2N$ strings $x_i$'s and $y_i$'s, each
comprising $n$ variables (over $\Sigma_k$), so that $m=Nn$.  We can
set $B{=}\frac{1}{Nn}$, since changing one string position changes the
sum of the $N$ edit distances by at most 1, and the quantity $1/(Nn)$
times the sum by at most $1/(Nn)$.  Therefore, considering that
$mB^2=(Nn)(1/(Nn)^2)=1/(Nn)$, we can write:
\begin{align}\label{eq:alpha-McDiarmid>}
  \Pr[\tilde{\alpha}_k(n,N)-\alpha_k(n) > \Delta] \leq
  \exp{\left(-Nn\Delta^2\right)}.
\end{align}
Symmetrically, it can be shown that
\begin{align}\label{eq:alpha-McDiarmid<}
  \Pr[\tilde{\alpha}_k(n,N)-\alpha_k(n) < \Delta] \leq
  \exp{\left(-Nn\Delta^2\right)}.
\end{align}
Combining Inequalities~(\ref{eq:alpha-McDiarmid>})
and~(\ref{eq:alpha-McDiarmid<}), after simple algebra, yields
Inequality~(\ref{eq:alpha-McDiarmid}).
\end{proof}

\begin{table*}[htbp]
    \centering
    \caption{Estimates $\tilde{e}_4(n,N)$ of the average edit distance
      and $\tilde{\alpha}_4(n,N)=\frac{1}{n}\tilde{e}_4(n,N)$ of the
      average distance per symbol $\alpha_4(n)$, for various string lengths $n$,
      based on $N{=}2^{39}/n^2$ samples.  The last column shows the
      confidence intervals for $\alpha_4(n)$ corresponding to confidence
      level $99.9\%$. $\tilde{S}_4(n,N)$ is the sample standard
      deviation of the single-pair distance.}
    \label{tab:mc}
    \vspace*{3mm}
    \begin{tabular}{ccccccc}
    \hline
    \\[-3mm]
    $n$ & $N$ & $\Delta_{99.9\%}(n,N)$ & $\tilde{e}_4(n,N)$ & $\tilde{S}_4(n,N)$ &
    $\tilde{\alpha}_4(n,N)$ & $99.9\%$ Conf. Int.\\
    \\[-3.5mm]
    \hline
    \\[-3mm]
$2^{8}$  & $2^{23}$ & $0.59~10^{-4}$ & $138.10$ & $3.838$ & $0.53946$ & $[0.53940, 0.53953]$\\
$2^{9}$  & $2^{21}$ & $0.84~10^{-4}$ & $272.10$ & $4.920$ & $0.53144$ & $[0.53135, 0.53153]$\\
$2^{10}$ & $2^{19}$ & $0.12~10^{-3}$ & $538.77$ & $6.307$ & $0.52614$ & $[0.52602, 0.52626]$\\
$2^{11}$ & $2^{17}$ & $0.17~10^{-3}$ & $1070.4$ & $8.146$ & $0.52263$ & $[0.52246, 0.52280]$\\
$2^{12}$ & $2^{15}$ & $0.24~10^{-3}$ & $2131.5$ & $10.56$ & $0.52039$ & $[0.52015, 0.52063]$\\
$2^{13}$ & $2^{13}$ & $0.34~10^{-3}$ & $4250.9$ & $13.62$ & $0.51891$ & $[0.51857, 0.51925]$\\
$2^{14}$ & $2^{11}$ & $0.48~10^{-3}$ & $8487.0$ & $17.71$ & $0.51801$ & $[0.51753, 0.51849]$\\
$2^{15}$ & $2^{9}$  & $0.67~10^{-3}$ & $16954$  & $24.98$ & $0.51739$ & $[0.51671, 0.51807]$\\
$2^{16}$ & $2^{7}$  & $0.95~10^{-3}$ & $33884$  & $29.12$ & $0.51704$ & $[0.51608, 0.51799]$\\
$2^{17}$ & $2^{5}$  & $0.13~10^{-2}$ & $67734$  & $38.85$ & $0.51677$ & $[0.51542, 0.51812]$\\
$2^{18}$ & $2^{3}$  & $0.19~10^{-2}$ & $135450$ & $62.94$ & $0.51670$ & $[0.51479, 0.51861]$\\
\hline
\end{tabular} 
\end{table*}


\paragraph{Remark}
Propositions \ref{prop:McD89} and \ref{prop:aplha:k:n:confidence} only require that the symbols of the random strings are statistically independent, not necessarily with the same distribution. The numerical results reported next refer to the special case where all symbols are uniformly distributed, which underlies the definition of $\alpha_k(n)$ and $e_k(n)$. However, the approach could be straightforwardly applied to other, possibly position-dependent, distributions. Interestingly, for given $n$ and $N$, the width of the confidence interval, for a given confidence level (\emph{e.g.}, $\Delta_{99.9\%}(n,N)$ in Table~\ref{tab:mc}), is independent of the distribution.

\paragraph*{Numerical results}
Table~\ref{tab:mc} reports estimates based on Monte Carlo estimates,
within the framework of Proposition~\ref{prop:aplha:k:n:confidence}.
The alphabet size is $k=4$, a case of special interest in DNA
analysis (\emph{e.g.}, \cite{GMR16, BPRS21}). For the string length $n$, the values
considered are the powers of two from $2^8=256$ to
$2^{18}=262144$.  For each $n$, the number of samples $N$ has been
chosen as $N=2^{39}/n^2$. This choice (roughly) equalizes the amount
of (sequential) computation time devoted to each $n$, when the edit
distance for a sample pair is computed by a quadratic algorithm, say,
in time $T_{ed}(n)=\tau_{ed}n^2$, for some constant $\tau_{ed}$,
making the overall time for $N$ samples $T(n,N)=\tau_{ed}Nn^2$. In
our experiments, this becomes $T(n,2^{39}/n^2)=\tau_{ed}2^{39}
\approx 2560s \approx 43m$, where we measured
$\tau_{ed}\approx1.25~2^{-28}s \approx 4.66 ns$, on a state-of-the-art processor core.
From Inequality~(\ref{eq:alpha-McDiarmid}), straightforward
manipulations show that, if the target is a confidence level
$\lambda$, then the radius $\Delta$ of the confidence interval
becomes:
\begin{align}
\label{eq:alpha(n):confidence:radius}
\Delta_{\lambda}(n,N) = \sqrt{\frac{1}{Nn} \ln \left(\frac{2}{1{-}\lambda}\right)}.
\end{align}
Choosing $\lambda{=}0.999$ and recalling that, in our experiments, we
have set $N{=}2^{39}/n^2$, the above formula becomes
\begin{align}
\label{eq:confidence:radius:exper}
  \Delta_{99.9\%}(n,2^{39}/n^2)=\sqrt{2^{-39}n \ln \left(2000\right)}
        =2^{-20} 3.90 \sqrt{n}.
\end{align}
Thus, for $n=2^8$, we have $\Delta_{99.9\%}=0.59~10^{-4}$. For
$n=2^{18}$, we have $\Delta_{99.9\%}=0.19~10^{-2}$. The values of
$\Delta_{99.9\%}$ have been used, together with the experimental
values of $\tilde{\alpha}_4(n,N)$, to obtain the confidence intervals
reported in the last column of Table~\ref{tab:mc}, as
$[\tilde{\alpha}_4(n,N)-\Delta_{99.9\%}(n,N),\tilde{\alpha}_4(n,N)+
  \Delta_{99.9\%}(n,N)]$.
  
\begin{table*}[tbp]
    \centering
    \caption{Estimates $\tilde{e}_k(n,N)$ of the average edit distance
      and $\tilde{\alpha}_k(n,N)=\frac{1}{n}\tilde{e}_k(n,N)$ of the
      average distance per symbol $\alpha_k(n)$, for various alphabet sizes $k$,
      based on $N=2^{9}$ samples of pairs of strings with length
      $n=2^{15}$.  The confidence radius is
      $\Delta_{99.9\%}(n,N)=\Delta_{99.9\%}(2^{15}, 2^{9}) \approx
      0.67~10^{-3}$; the corresponding confidence intervals for $\alpha_k(n)$ are given
      in the last column. $\tilde{S}_k(n,N)$ is the
      sample standard deviation of the single-pair distance.}
    \label{tab:mc:k}
    \vspace*{3mm}
    \begin{tabular}{ccccc}
    \hline
    \\[-3mm]
    $k$ & $\tilde{e}_k(n,N)$ & $\tilde{S}_k(n,N)$ &
    $\tilde{\alpha}_k(n,N)$ & $99.9\%$ Conf. Int.\\
    \\[-3.5mm]
    \hline
    \\[-3mm]
$2$  & $9442.6$ & $26.04$ & $0.28817$ & $[0.28749, 0.28884]$\\
$3$  & $14042$  & $24.88$ & $0.42852$ & $[0.42784, 0.42920]$\\
$4$  & $16954$  & $24.98$ & $0.51739$ & $[0.51671, 0.51807]$\\
$5$  & $19005$  & $22.78$ & $0.57998$ & $[0.57930, 0.58066]$\\
$6$  & $20549$  & $22.00$ & $0.62710$ & $[0.62642, 0.62778]$\\
$7$  & $21761$  & $21.05$ & $0.66409$ & $[0.66341, 0.66477]$\\
$8$  & $22742$  & $20.15$ & $0.69402$ & $[0.69334, 0.69470]$\\
$16$ & $26839$  & $16.38$ & $0.81906$ & $[0.81838, 0.81974]$\\
$32$ & $29471$  & $14.10$ & $0.89939$ & $[0.89871, 0.90007]$\\
\hline
\end{tabular} 
\end{table*}

Table~\ref{tab:mc:k} reports estimates of $\alpha_k(n)$ for various
alphabet sizes $k$. The confidence intervals shown in the last column
of the table are based on the confidence level $\lambda=0.999$. The
estimates are obtained from numerical experiments involving $N=2^{9}$
random pairs of strings of length $n=2^{15}$.  Since the radius of the
confidence interval only depends on $\lambda$, $n$, and $N$ (but not
$k$), the same value $\Delta_{99.9\%}(2^{15},2^{9})\approx
0.67~10^{-3}$ applies to all $k$'s.  Notice that the standard deviation
$\tilde{S}_k(n,N)$ tends to decrease with $k$; this makes intuitive
sense since, for fixed $n$, as $k$ increases, the probability that every symbol that appears in one string is distinct from every symbol that
appears in the other string approaches $1$, so that almost all
pairs of strings have distance $n$, hence the variance is negligible.

\subsection{$\alpha_k$}
\label{sec:alpha:k}    
The following proposition provides confidence intervals for
$\alpha_k$, when the quantity $\tilde{\alpha}_k(n,N)-Q(n)/2$ is used
as an estimator, where $Q(n)$ bounds the rate of convergence,
according to Equation~(\ref{eq:rate:alpha:Q}).
\begin{proposition}[Confidence intervals for $\alpha_k$]
\label{prop:alpha:k:confidence}
For any $\Delta \geq 0$, the (random) interval \emph{centered} at
$\tilde{\alpha}_{k}(n,N)-Q(n)/2$ and of \emph{radius}
$\Delta+Q(n)/2$, \emph{i.e.},
\begin{align*}
  I_k&(n,N,\Delta) =\\ &[(\tilde{\alpha}_{k}(n,N)-Q(n)/2){-}(\Delta+Q(n)/2),
    (\tilde{\alpha}_{k}(n,N)-Q(n)/2){+}(\Delta+Q(n)/2)]
\end{align*}
is a \emph{confidence interval} for $\alpha_{k}$, with
\emph{confidence level} $1-2\exp{(-Nn\Delta^2)}$,
\emph{i.e.},
\begin{align}\label{eq:alpha-inf-McDiarmid}
  \Pr[\alpha_k \in I_k(n,N,\Delta)] \geq 1-2\exp{\left(-Nn\Delta^2\right)}.
\end{align}
\end{proposition}
\begin{proof}
By Theorem~\ref{thm:computability:alpha} and the fact that
$\alpha_k < \alpha_k(n)$, we have that
$\alpha_k(n)-Q(n)\leq\alpha_k\leq\alpha_k(n)$. By
Inequality~(\ref{eq:alpha-McDiarmid}), this implies
\begin{align*}
  \Pr[\tilde{\alpha}_{k}(n,N)-Q(n)-\Delta \leq \alpha_k \leq
    \tilde{\alpha}_{k}(n,N)+\Delta] \geq 1-2\exp{\left(-Nn\Delta^2\right)}.
\end{align*}
To arrive at Inequality~(\ref{eq:alpha-inf-McDiarmid}), it remains to
observe that the interval within the probability is just a rewriting
of $I_k(n,N,\Delta)$.
\end{proof}

From Proposition~\ref{prop:alpha:k:confidence}, straightforward
manipulations show that, if the target confidence level is
$\lambda$, then the radius $R_{\lambda}(n,N)$ of the confidence
interval becomes:
\begin{align}\label{eq:alpha:confidence:radius}
  R_{\lambda}(n,N) = \frac{Q(n)}{2}{+}
  \sqrt{\frac{1}{Nn} \ln \left(\frac{2}{1{-}\lambda}\right)}.
\end{align}
The first term arises because $\alpha_k(n)$ is a (deterministically
bounded) approximation to $\alpha_k$ and the second term because
$\tilde{\alpha}_k(n,N)$ is a (statistically bounded) approximation to
$\alpha_k(n)$. To get a sense of the relative weight of the two terms
contributing to the radius, we observe that, from
Equation~(\ref{eq:rate:alpha:Q}), we can approximate the first term as
$\frac{Q(n)}{2} \approx \sqrt{\frac{\ln n}{2n}}$.  Therefore, even
using just one sample pair ($N=1$), this term dominates the second
one as soon as $n>\left(\frac{2}{1-\lambda} \right)^2$. For
$\lambda=0.999$, we get $n>4~10^{6}$.

\paragraph{remark}
It is interesting to observe that the size
$2R_{\lambda}(n,N)$ of the confidence interval is independent of the
alphabet size $k$. The same applies to the computational work,
$T(n,N)=\tau_{ed}Nn^2+l.o.t.$, to obtain the statistical estimate
of $\alpha_k(n)$ with a given confidence.

\begin{table*}[htbp]
    \centering
    \caption{Estimates $\hat{\alpha}_4(n,N)$ of $\alpha_4$ with
      confidence intervals for various string lengths $n$, based on
      $N=2^{39}/n^2$ samples (so that each estimate takes
      approximately the same computation time).  $I_4(n,N,\Delta)$ is
      the corresponding $99.9\%$ confidence interval for
      $\alpha_4$. The estimates improve with growing $n$, as witnessed
      by the decrease of the interval radius $R_{0.999}(n,N)$.}
    \label{tab:hat:alpha}
    \vspace*{3mm}
    \begin{tabular}{ccccc}
    \hline
    \\[-3mm]
    $n$ & $N$ & $R_{0.999}(n,N)$ &
    $\hat{\alpha}_4(n,N)$ & $I_4(n,N,\Delta)$\\
    \\[-3mm]
    \hline
    \\[-3mm]
$2^{8}$  & $2^{23}$ & $0.11534$ & $0.42418$ & $[0.30884, 0.53953]$\\
$2^{9}$  & $2^{21}$ & $0.08523$ & $0.44629$ & $[0.36105, 0.53153]$\\
$2^{10}$ & $2^{19}$ & $0.06287$ & $0.46338$ & $[0.40051, 0.52626]$\\
$2^{11}$ & $2^{17}$ & $0.04631$ & $0.47649$ & $[0.43075, 0.52280]$\\
$2^{12}$ & $2^{15}$ & $0.03409$ & $0.48654$ & $[0.45244, 0.52063]$\\
$2^{13}$ & $2^{13}$ & $0.02512$ & $0.49413$ & $[0.46900, 0.51925]$\\
$2^{14}$ & $2^{11}$ & $0.01858$ & $0.49990$ & $[0.48132, 0.51849]$\\
$2^{15}$ & $2^{9}$  & $0.01388$ & $0.50419$ & $[0.49031, 0.51807]$\\
$2^{16}$ & $2^{7}$  & $0.01056$ & $0.50742$ & $[0.49685, 0.51799]$\\
$2^{17}$ & $2^{5}$  & $0.00833$ & $0.50979$ & $[0.50145, 0.51812]$\\
$2^{18}$ & $2^{3}$  & $0.00698$ & $0.51163$ & $[0.50465, 0.51861]$\\
\hline
\end{tabular} 
\end{table*}

Our estimate of $\alpha_k$ is the center of the interval
$I_k(n,N,\Delta)$ as defined in
Proposition~\ref{prop:alpha:k:confidence}
\begin{align}
\label{eq:alpha:estimates}
\hat{\alpha}_k(n,N,\Delta) = \tilde{\alpha}_{k}(n,N)-Q(n)/2.
\end{align}
Table~\ref{tab:hat:alpha} reports these estimates for $k=4$, based on the
values $\tilde{\alpha}_4(n,N)$ reported in Table~\ref{tab:mc} and on
Proposition~\ref{prop:alpha:k:confidence} for confidence level
$\lambda=0.999$.

Table~\ref{tab:hat:alpha:k} reports estimates of $\alpha_k$, for
various values of $k$, all based on numerical experiments with
$N=2^{9}$ samples of string pairs of size $n=2^{15}$.  The
corresponding confidence interval $I_k(n,N,\Delta)$ is obtained at
confidence level $\lambda=0.999$. All intervals have radius
$R_{0.999}(2^{15},2^{9}) \approx 1.4~10^{-2}$.

\begin{table*}[tbp]
    \centering
    \caption{Estimates of $\alpha_k$ for various alphabet sizes $k$,
      based on $N=2^{9}$ samples of pairs of string with length
      $n=2^{15}$. The radius of the interval $I_k(n,N,\Delta)$ is
      based on a $\lambda=0.999$ confidence level and is
      $R_{0.999}(2^{15},2^{9}) \approx 1.4~10^{-2}$, for all values of
    $k$.}
    \label{tab:hat:alpha:k}
    \vspace*{3mm}
    \begin{tabular}{ccc}
    \hline
    \\[-3mm]
    $k$ & $\hat{\alpha}_k(n,N)$ & $I_k(n,N,\Delta)$\\
    \\[-3mm]
    \hline
    \\[-3mm]
    $2$    & $0.27496$ & $[0.26108, 0.28884]$\\
    $3$    & $0.41532$ & $[0.40144, 0.42920]$\\
    $4$    & $0.50419$ & $[0.49031, 0.51807]$\\
    $5$    & $0.56678$ & $[0.55289, 0.58066]$\\
    $6$    & $0.61390$ & $[0.60002, 0.62778]$\\
    $7$    & $0.65089$ & $[0.63701, 0.66477]$\\
    $8$    & $0.68082$ & $[0.66694, 0.69459]$\\
    $16$   & $0.80586$ & $[0.79198, 0.81974]$\\
    $32$   & $0.88619$ & $[0.87230, 0.90007]$\\
    \hline
    \end{tabular}
\end{table*}

\section{Upper bounds for $\alpha_k$}
\label{sec:upperbound}
In this section, we present methods to derive upper bounds to $\alpha_k$
based on the exact computation of $\alpha_k(n) = e_k(n)/n$ for some
$n$, and on the relation $\alpha_k \leq \alpha_k(n)$, valid for all $n
\geq 1$.  The computation of $e_k(n)$ can be reduced to that of the
eccentricity, as in Equation (\ref{eq:ed:eccentricty}) repeated here
for convenience:
\begin{align} \label{eq:ed:eccentricty2}
e_k(n) = k^{-n}\sum_{x\in\Sigma_k^n}{\ecc(x)}.
\end{align}
If $\ecc(x)$ is computed according to Equation (\ref{eq:eccentricity})
and the distance $d_E(x,y)$ is computed by the $\bigO(n^2)$-time
dynamic programming algorithm for each of the $k^n$ strings $y
\in\Sigma_k^n$, then the overall computation time is $\bigO(n^2
k^{n})$ for $\ecc(x)$ and $\bigO(n^2 k^{2n})$ for $e_k(n)$, since the
eccentricity of each of the $k^n$ strings $x \in \Sigma_k^n$ is needed
in Equation (\ref{eq:ed:eccentricty2}).  Below, we propose a more
efficient algorithm to speed up the computation of $\ecc(x)$ and, in
turn, that of $e_k(n)$, achieving time $\bigO(n^2\min{(k,3)}^n
k^n)=\bigO(n^2 3^n k^n)$.  We also show how to exploit some symmetries
of $\ecc(x)$ in order to limit the computation of the eccentricity
needed to obtain $e_k(n)$ to a suitable subset of $\Sigma_k^n$.

\subsection{The coalesced dynamic programming algorithm for
  eccentricity} Let $\mathbf{M}(x,y)$ be the matrix produced by the
dynamic programming algorithm (reviewed in Section
\ref{sec:DPalgorithms}) to compute $d_E(x,y)$, with
$x,y\in\Sigma_k^n$. We develop a strategy to coalesce the computations
of $\mathbf{M}(x,y)$ for different $y \in\Sigma_k^n$, while keeping
$x$ fixed.  To this end, we choose to generate the entries of
$\mathbf{M}(x,y)$, according to Equation
(\ref{eq:edit:distance:recurrence}), in column-major order.  Clearly,
column $j$ is fully determined by $x$ and by the prefix of $y$ of
length $j$.  Define now the \emph{column multiset} $\mathcal{C}_j$
containing column $j$ (\emph{i.e.}, the
last one) of $\mathbf{M}(x,y[1] \ldots y[j])$ for each string $y[1]
\ldots y[j] \in \Sigma_k^j$. The multiset $\mathcal{C}_j$ is a
function of (just) $x$, although, for simplicity, the dependence upon
$x$ is not reflected in our notation. Clearly, each column in
$\mathcal{C}_{j-1}$ generates $k$ columns in $\mathcal{C}_j$, one for
each symbol of $\Sigma_k$; therefore the cardinality of
$\mathcal{C}_j$ is $|\mathcal{C}_j|=k^j$.  However, several columns
may be equal to each other, so that the number of distinct such
columns can be much smaller. In fact, we will show that this number is
upper bounded by $3^n$, which is smaller than $k^j$ for
$j{>}\frac{\log_2 3}{\log_2 k}n$.  This circumstance can be exploited
to save both space and computation, by representing $\mathcal{C}_j$ as
a set of records each containing a distinct column and its
multiplicity. Intuitively, we are coalescing the computation of the
dynamic programming matrices corresponding to different strings, when
such matrices happen to have the same $j$-th column.

\begin{algorithm}[tb]
\begin{algorithmic}[1]
\Procedure{Eccentricity}{$x$}
  \State $n \gets |x|$
  \State $\mathcal{C}_0 \gets \{((0,1,\ldots,n), 1)\}$
  \For{$j\gets 1$ \textbf{to} $n$}
    \State $\mathcal{C}_j \gets \emptyset$
    \For{($\mathbf{c}, \mu(\mathbf{c})) \in \mathcal{C}_{j-1}$}
      \For{$b\in \Sigma_k$}
	\State $\mathbf{c}' \gets$ \Call{NextColumn}{$x,\mathbf{c},j,b$} 
        \label{alg:ecc:trans}
	\State \Call{Insert}{$\mathcal{C}_j,(\mathbf{c}',\mu(\mathbf{c}))$} 
        \label{alg:ecc:add}
      \EndFor
    \EndFor
  \EndFor
  \State $e \gets 0$
  \For{$\mathbf{c}\in\mathcal{C}_n$}
    \State $e \gets e + \mu(\mathbf{c})*\mathbf{c}[n]$
  \EndFor
  \State \Return{$e/k^n$}
\EndProcedure
\end{algorithmic}
\caption{Coalesced dynamic programming algorithm to compute $\ecc(x)$} \label{alg:ecc}
\end{algorithm}

The \emph{Coalesced Dynamic Programming} (CDP) algorithm described
next (referring also to the line numbers of Algorithm \ref{alg:ecc}),
constructs the sequence of multisets $\mathcal{C}_0, \mathcal{C}_1,
\ldots, \mathcal{C}_n$. A column multiset $\mathcal{C}$ will be
represented as a set of pairs $(\mathbf{c},\mu(\mathbf{c}))$, one for
each distinct member $\mathbf{c}$, with $\mu(\mathbf{c})$ being the
multiplicity of $\mathbf{c}$ in $\mathcal{C}$. The eccentricity of $x$
is obtained (lines 13-17) as the weighted average of the $n$-th
elements of all columns in $\mathcal{C}_n$:
\begin{align}
\label{eq:CDPoutput}
\ecc(x) = k^{-n} \sum_{\mathbf{c} \in \mathcal{C}_n} {\mu(\mathbf{c})
  \mathbf{c}[n]}.
\end{align}
As can be seen from Equation (\ref{eq:edit:distance:recurrence}),
multiset $\mathcal{C}_0$ contains the column $(0,1,\ldots,n)$,
with multiplicity 1 (line 3).  For $j=1, \ldots, n$, $\mathcal{C}_j$
is obtained by scanning all $\mathbf{c} \in \mathcal{C}_{j-1}$ (line
6) and all $b \in \Sigma_k$ (line 7), and by
\begin{itemize}
\item computing the $j$-th column $\mathbf{c}'$ resulting from
  Equation (\ref{eq:edit:distance:recurrence}) when the $(j-1)$-st column
  is $\mathbf{c}$ and $\xi_{i,j}=0$ if $x[i]=b$ or else $\xi_{i,j}=1$
  (call to \Call{NextColumn}{$x,\mathbf{c},j,b$}, line 8);
\item inserting $\mu(\mathbf{c})$ copies of $\mathbf{c}'$ in
  $\mathcal{C}_j$, by either creating a new pair
  $(\mathbf{c}',\mu(\mathbf{c}))$ when $\mathbf{c}'$ is not present in
  the multiset or by incrementing its multiplicity by
  $\mu(\mathbf{c})$ otherwise (call to
  \Call{Insert}{$\mathcal{C}_j,(\mathbf{c}',\mu(\mathbf{c}))$}, line
  9).
\end{itemize}
The correctness of the CDP algorithm is pretty straightforward to
establish.  A few observations are however necessary in order to
describe and analyze a data structure that can efficiently implement,
in our specific context, multisets with the insertion operation.  The
key property is that, for $j=0,1\ldots,n$, the column of
$\mathbf{M}(x,y)$ with index $j$ satisfies the conditions (a)
$M_{0,j}=j$ and (b) $(M_{i,j}-M_{i-1,j}) \in \{-1,0,1\}$, for $i=1,
\ldots, n$.  Using this property, the set of distinct columns that
belong to the multiset $\mathcal{C}_j$ can be represented as a ternary
tree where each arc has a label from the set $\{-1,0,1\}$ and a column
$(M_{0,j},M_{1,j},\ldots,M_{n,j})$ is mapped to a leaf $v$ such that
the $n$ arcs in the path from the root to $v$ have labels
$(M_{1,j}-M_{0,j}),\ldots, (M_{n,j}-M_{n-1,j})$.  Each leaf stores the
multiplicity of the corresponding column.  The size of the tree for
$\mathcal{C}_j$ is $\bigO(\min(3^n,k^j))$, since there are at most
$3^n$ columns satisfying the constraints and $k^j$ $k$-ary strings that
contribute (not necessarily distinct) columns. Hence, the body of the
loop, whose iteration range is defined in lines 4, 6, and 7, is executed
$n k \bigO(\min(3^n,k^j))$ times.  Considering that one call to
\textsc{NextColumn()} as well as one call to \textsc{Insert()} can be
easily performed in $\bigO(n)$ time, we can summarize the previous
discussion as follows, where we also consider that, at any given time,
the algorithm only needs to store two consecutive column multisets.
\begin{proposition}
\label{thm:ed:CDPtimeSpace}
The \emph{CDP algorithm} computes the \emph{eccentricity} $\ecc(x)$ of a
string $x$ of length $n$ over a $k$-ary alphabet in time
$T=\bigO(n^2k\min(3^n,k^n))$ and space $S=\bigO(\min(3^n,k^n))$.
Correspondingly, the \emph{expected distance} $e_k(n)$ can be computed
in time $T=\bigO(n^2k^{n+1}\min(3^n,k^n))$ and space
$S=\bigO(\min(3^n,k^n))$.
\end{proposition}

\paragraph{remark}
The CDP algorithm can be easily generalized to handle the case where $\ecc(x)$ is defined with respect to a random string $y$ whose symbols are independently, but not necessarily uniformly, distributed (the probability of the string $y[1]\ldots y[n]$ has the form $\prod_{j=1}^{n}{p_j(y[j])}$). Essentially, rather than maintaining the multiplicity of each column $\mu(\mathbf{c})$, we maintain its probability $\pi(\mathbf{c})$. The rules to obtain $\mathcal{C}_{j}$ from $\mathcal{C}_{j-1}$ are straightforward. The eccentricity is obtained as
$\sum_{\mathbf{c} \in \mathcal{C}_n}{\pi(\mathbf{c})\mathbf{c}[n]}$. This generalized version of the CDP algorithm naturally enables the computation of the expected edit distance between two random strings $x$ and $y$ whose symbols are all mutually independent. The computational bounds remain those stated in Proposition \ref{thm:ed:CDPtimeSpace}.

\subsection{Exploiting symmetries of $\ecc(x)$ in the computation of
  $e_k(n)$}
The edit distance enjoys some useful symmetries, which can be easily
derived from the definition. One is that, if we let $x^R=x[n] \ldots
x[1]$ denote the \emph {reverse} of the string $x=x[1] \ldots x[n]$, then
$d_E(x,y)=d_E(x^R, y^R)$. Another one is that if $\pi:~\Sigma_k
\rightarrow \Sigma_k$ is a permutation of the alphabet symbols and
$\pi(x)$ denotes the string $\pi(x[1]) \ldots \pi(x[n])$, then
$d_E(x,y)=d_E(\pi(x),\pi(y))$.  The following is a simple, but useful
corollary of these properties.

\begin{proposition}
\label{thm:ed:symmetry}
For any $x \in \Sigma_k^n$, we have $\ecc(x^R)=\ecc(x)$. Furthermore,
for any permutation $\pi$ of $\Sigma_k$, we have
$\ecc(\pi(x))=\ecc(x)$.
\end{proposition}
It is useful to define the equivalence class of $x$ as the set of
strings that have the same eccentricity as $x$, due to Proposition
\ref{thm:ed:symmetry}, and denote by $\nu(x)$ the cardinality of such
set.  If $\mathcal{R}_{k,n} \subseteq \Sigma_k^n$ contains exactly one
(representative) member for each equivalence class, then Equation
(\ref{eq:ed:eccentricty2}) can be rewritten as
\begin{align} \label{eq:ed:eccentricty3}
e_k(n) = k^{-n}\sum_{x\in\mathcal{R}_{k,n}}{\nu(x)}{\ecc(x)}.
\end{align}
Computing $e_k(n)$ according to Equation (\ref{eq:ed:eccentricty3})
enables one to reduce the number of strings for which the eccentricity
has to be computed (via the CDP algorithm) by a factor slightly
smaller than $(2k!)$, with a practically appreciable reduction in
computation time.

The strategy outlined in this section has been implemented in C++ and
run on a $32$ core IBM Power7 server. For several alphabet sizes $k$,
we have considered values of $n$ up to a maximum value $n^{ub}_k$,
under the constraint that the running time would not exceed one
week. The resulting values $e_k(n^{ub}_k)$ are presented 
in Table~\ref{tab:upper:bounds}.  For the quaternary alphabet, we
obtain $\alpha_4 \leq 0.6318$, which is rather loose because based
on a small string length, namely $n_4^{ub}=15$. The limitation on the
string length is obviously due to the high complexity of the
algorithm. In contrast, the statistical estimate presented in
Section~\ref{sec:montecarlo} is based on much longer strings and
represents a more accurate approximation of $\alpha_4$, although the
estimate comes with a confidence interval, rather than with a
deterministic guarantee.

We mention that the CDP algorithm could perhaps be improved, by a
constant factor related to the machine word-size, with the bit-vector
approach presented in \cite{Mye99}, or by a logarithmic factor, with
the the Four Russians method in \cite{MP80}. Considering the
exponential nature of the CDP algorithm, however, these approaches are
unlikely to yield substantial improvements of the deterministic upper
bound and we have not pursued them.

\begin{table*}[t]
    \centering
    \caption{Values $\alpha_k(n)$ computed using the Coalesced Dynamic
      Programming algorithm for various alphabets. The string lengths
      $n_k^{ub}$'s have been chosen so that the total time needed to
      compute $\alpha_k(n)$ is roughly the same for each $k$
      (approximately one week on a machine with $32$ cores).}
    \label{tab:upper:bounds}
    \begin{tabular}{ccc}
    \hline
    \\[-3mm]
    $k$  & $n_k^{ub}$ & $\alpha_k(n_k^{ub})$\\
    \\[-3mm]
    \hline
    \\[-3mm]
    $2$  & $24$ & $0.36932$\\
    $3$  & $17$ & $0.53426$\\
    $4$  & $15$ & $0.63182$\\
    $5$  & $13$ & $0.70197$\\
    $6$  & $12$ & $0.75149$\\
    $7$  & $11$ & $0.79031$\\
    $8$  & $11$ & $0.81166$\\
    $16$ & $10$ & $0.89554$\\
    $32$ & $6$  & $0.96588$\\
    \hline
    \end{tabular}
\end{table*}

\section{Lower bounds for $\alpha_k$}
\label{sec:lowerbound}
In this section, we establish a lower bound $\beta_k^*$ to $\alpha_k$,
for each $k\geq 2$. We first characterize $\beta_k^*$ analytically, as
the supremum of a suitably defined set of real numbers, and then
provide an efficient algorithm to compute $\beta_k^*$, within any
desired approximation. We focus on the expression of $e_k(n)$ in terms
of eccentricities, given in the second line of
Equation~(\ref{eq:ed:eccentricty}). We derive lower bounds to
$\ecc(x)$ by ignoring the contribution of the strings inside the ball
of radius $r$ centered at $x$ and by setting to $r+1$ the contribution
of the string outside the same ball.  The objective is to determine
the largest value $r_k^*(n)$ of $r$ for which (it can be shown that)
the ball of radius $r$ contains a fraction of $\Sigma_k^n$ that
vanishes with $n$; then $r_k^*(n)/n$ will converge to a lower bound to
$\alpha_k$.  Below, we formalize this idea and show that we can choose
$r_k^*(n)=\beta_k n$ for suitable values of $\beta_k$ independent of
$n$; this establishes that $\alpha_k \geq \beta_k^*$, where
$\beta_k^*$ is the supremum of such values.
It is shown that $\lim_{k \to \infty} \beta_k^* =1$ and $\beta_k^*
\leq \alpha_k \leq 1-\frac{1}{k}$, whence, for $k$ large enough,
$\beta_k^*$ provides an increasingly accurate estimate of
$\alpha_k$. Thus, we turn our attention to translating the analytical
characterization of $\beta_k^*$ into an efficient numerical algorithm
for its computation, a translation which is not completely
straightforward. Finally, we present the numerical values of our lower
bound for a sample of alphabet sizes.

\subsection{Lower bounds to $\ecc(x)$ from
  upper bounds to ball size}
In this subsection, we derive lower bounds to $\ecc(x)$ based on upper
bounds to the size of the ball of radius $r$ centered at $x$.  Such
bounds hold for every string $x$, but depend only upon the length $n$
of $x$. They will be used to compute lower bounds to $\alpha_k$.
\begin{definition}
For a string $x \in \Sigma_k^n$, the \emph{ball} of radius $r$
centered at $x$ is defined as the set of strings having edit distance at
most $r$ from $x$:
\begin{align*}
B_{k,r}(x) = \{ y \in \Sigma_k^n : d_E(x,y) \leq r \} .
\end{align*}
Similarly, the \emph{shell} of radius $r$ centered at $x$ is defined
as the set of strings having edit distance exactly $r$ from $x$:
\begin{align*}
S_{k,r}(x) = \{ y \in \Sigma_k^n : d_E(x,y) = r \}.
\end{align*}
\end{definition}
The next lemma shows how an upper bound to the ball size can provide a
lower bound to the eccentricity, thus motivating the derivation of
such an upper bound.

\begin{lemma}
\label{thm:volume}
Let $u_{k,r}(x) \geq |B_{k,r}(x)|$; then for every $r^*=0,1,\ldots,n$:
\begin{align} \label{eq:lb:alpha:general}
 \ecc(x) &\geq r^*\left(1 - k^{-n}u_{k,r^*}(x)\right).
\end{align}
\end{lemma}
\begin{proof}
By partitioning $\Sigma_k^n$ into shells centered at $x$, we can rewrite (\ref{eq:eccentricity}) as
\begin{align*}
  \ecc(x) &= k^{-n}\sum_{r=0}^{r^*}{r|S_{k,r}(x)|} +
  k^{-n}\sum_{r=r^*+1}^{n}{r|S_{k,r}(x)|}\\
        &\geq k^{-n}(r^*+1)\sum_{r=r^*+1}^{n}{|S_{k,r}(x)|} \\
        &= k^{-n}(r^*+1)\left(|B_{k,n}(x)|-|B_{k,r^*}(x)| \right)\\
        &> r^*\left(1-k^{-n}|B_{k,r^*}(x)|\right) \\
  &\geq r^*\left(1-k^{-n}u_{k,r^*}(x)\right),
\end{align*}
where, in the last two steps, the relationships $|B_{k,n}(x)|=k^n$
and $|B_{k,r^*}(x)| \leq u_{k,r^*}(x)$ have been utilized.
\end{proof}

The bound $u_{k,r^*}(x)$ we derive below depends only upon the length
$n$ of $x$ so that it can be written as $u_{k,r^*}(n)$, with a
harmless overloading of notation. Then, simple manipulations of
Equation (\ref{eq:ed:eccentricty}) show that
\begin{align} \label{eq:alpha-k:lb}
  \alpha_k(n)=\frac{e_k(n)}{n}\geq\frac{r^*}{n}
   \left(1 - k^{-n}u_{k,r^*}(n)\right);\\
     \label{eq:alpha:lb}
\alpha_k = \lim_{n\rightarrow \infty} \alpha_k(n) \geq
\lim_{n\rightarrow\infty}{\frac{r^*}{n}\left(1-k^{-n}u_{k,r^*}(n)\right)}.
\end{align}
We will show that, for suitable values of $\beta_k$, letting
$r^*=\beta_k n$, the quantity $k^{-n}u_{k,\beta_k n}(n)$ converges to
$0$ whence, by Equation (\ref{eq:alpha:lb}), $\alpha_k \geq \beta_k$.

\subsection{Upper bounds on ball size}
To apply Lemma \ref{thm:volume}, we need an upper bound to
$|B_{k,r}(x)|$.  The next proposition develops such an upper bound by
(i) showing that every string $y \in B_{k,r}(x)$ can be obtained from
$x$ by applying a script of certain type with cost $r$ or $r-1$ and
(ii) counting such scripts.  In general, the upper bound will not be
tight, because the count may include multiple scripts that
produce $y$ from $x$.

\begin{proposition}
\label{thm:scripts}
For any $x\in\Sigma_k^n$ and for any $r=1,\ldots,n$
\begin{align} \label{eq:script:bound}
|B_{k,r}(x)| \leq (k-1)^r \sum_{d=0}^{\lfloor r/2 \rfloor}{
  \binom{n}{d}^2\binom{n-d+1}{r-2d}\left(\frac{k}{(k-1)^2}\right)^d} .
\end{align}
\end{proposition}
\begin{proof}
We call \emph{canonical simple script (CSS)} a simple script (see
Section~\ref{sec:notationsAndDef}) where all deletions precede all
substitutions, the latter precede all insertions and, within each
type of operation, cells are processed from left to right. For any
script transforming $x$ into $y$, there is a CSS of non greater cost
which achieves the same transformation. Therefore, if $d_E(x,y)=r$,
then there is a CSS of cost $r$ which, applied to $x$, produces
$y$. Each CSS of cost $r\in\{0,1,\ldots,n\}$ can be constructed by
a sequence of choices, as specified below (shown within
square brackets is the number of possible choices):
\begin{itemize}
    \item $d\in\{0,1,\ldots, \lfloor r/2 \rfloor\}$
    \item $d$ positions to delete from $x$ [$\binom{n}{d}$]
    \item $(r-2d)$ of the remaining $(n-d)$ positions to
        be substituted [$\binom{n-d}{r-2d}$]
    \item $d$ positions to insert in $y$ [$\binom{n}{d}$]
    \item the symbols in the substitutions [$(k-1)^{r-2d}$]
    \item the symbols in the insertions [$k^d$]
\end{itemize}
Straightforwardly, the number of CSSs of cost $r$ is
\begin{align} \label{eq:cost:r:scripts}
s_{k,r} = \sum_{d=0}^{\lfloor r/2 \rfloor}{\binom{n}{d}^2\binom{n-d}{r-2d}}
     (k-1)^{r-2d}k^d .
\end{align}

Next, we prove that any $y\in B_{k,r}(x)$ can be obtained from $x$ via
a simple script of cost $r-1$ or $r$. Let $r'=d_E(x,y)\leq r$.
When $r'=r$, an optimal script of cost $r'$ is also a simple script of
the same cost. Hence, the canonical version of such optimal
script can be used to obtain $y$ from $x$. The same reasoning applies
to the case $r'=r-1$. Finally, for $r'<r-1$, consider an optimal CCS
of cost $r'$ that transforms $x$ into $y$.  By augmenting this script
with $\lfloor (r-r')/2 \rfloor$ pairs of deletions and insertions,
each pair acting on a matched position, we obtain a simple script of
cost $r$, if $r-r'$ is even, or of cost $r-1$ if $r-r'$ is odd.  The
prescribed augmentation is always possible since the number of matches
is at least $n-r'\geq r-r' \geq (r-r')/2$.

The thesis is then established by the following chain of inequalities:
\begin{align*}
|B_{k,r}&(x)| \leq s_{k,r} + s_{k,r-1}\\
\leq &\sum_{d=0}^{\lfloor r/2 \rfloor}{\binom{n}{d}^2\binom{n-d}{r-2d}}
     (k-1)^{r-2d}k^d +
\sum_{d=0}^{\lfloor (r-1)/2 \rfloor}{\binom{n}{d}^2\binom{n-d}{r-1-2d}}
     (k-1)^{r-1-2d}k^d~~~~~ \\
\leq&
\sum_{d=0}^{\lfloor r/2 \rfloor}{\binom{n}{d}^2\binom{n-d}{r-2d}}
     (k-1)^{r-2d}k^d +
\sum_{d=0}^{\lfloor r/2 \rfloor}{\binom{n}{d}^2\binom{n-d}{r-1-2d}}
     (k-1)^{r-2d}k^d  \\
=&
\sum_{d=0}^{\lfloor r/2 \rfloor}{\binom{n}{d}^2\binom{n-d+1}{r-2d}}
     (k-1)^{r-2d}k^d,
\end{align*}
where we have made use of the identity
\begin{align*}
\binom{n-d}{r-2d} + \binom{n-d}{r-1-2d} = \binom{n-d+1}{r-2d}.
\end{align*} 
\end{proof}

\subsection{Asymptotic behavior of ball size and bounds for $\alpha_k$}
The next results show that the right hand side of
Inequality~(\ref{eq:script:bound}), divided by $k^n$, is bounded by a
sum of exponential functions whose exponents all vanish with $n$, when
the ball radius is set to $\beta_k n$, with $\beta_k$ satisfying
certain conditions (depending upon $k$). Intuitively, this means that,
except for a vanishing fraction, all strings in $\Sigma_k^n$ lie
outside of the ball $B_{k,\beta_k n}(x)$ whence, by Equation
(\ref{eq:alpha:lb}), $\alpha_k\geq \beta_k$.


\begin{definition}\label{def:H(x)}
Let $H(x)$, with $0 \leq x \leq 1$, denote the \emph{binary entropy
function}
\begin{align*}
H(x) = -x \log_2{x} - (1-x)\log_2{(1-x)},
\end{align*}
and let $H'(x)=\frac{dH}{dx}=\log_2\left(\frac{1-x}{x}\right)$ and
  $H''(x)=\frac{d^2H}{d^2x}=-\frac{\log_2 e}{x(1-x)}$ denote its first
  and second derivatives.
\end{definition}

\begin{definition}
For $\beta \in [0,1]$ and $\delta \in [0,\beta/2]$, we define the
function
\begin{align} \label{eq:g:exponent}
g_k(\beta,\delta) 
    =& (\beta-2\delta)\log_2{(k-1)}-(1-\delta)\log_2{k}\nonumber \\
    &+2H(\delta) + (1-\delta)H\left(\frac{\beta-2\delta}{1-\delta}\right).
\end{align}
\end{definition}
\begin{lemma}
\label{thm:exponential}
Let $u_{k,r}(n)$ be given by the right hand side of
(\ref{eq:script:bound}) and $g_k(\beta,\delta)$ be given by
(\ref{eq:g:exponent}). For every $\beta\in[0,1]$,
\begin{align} \label{eq:exponential:bound}
k^{-n}u_{k,\beta n}(n) \leq (n+1)\sum_{d=0}^{\lfloor \beta n/2 \rfloor}
    {2^{n g_k\left(\beta,\frac{d}{n}\right)}}.
\end{align}
\end{lemma}
\begin{proof}
Using the relation
\begin{align*}
\binom{n-d+1}{r-2d} = \frac{n-d+1}{n-r+d+1}\binom{n-d}{r-2d}
    \leq (n+1)\binom{n-d}{r-2d},
\end{align*}
the bound $\binom{n}{k}\leq
2^{nH(k/n)}$ (see, \emph{e.g.}, Eq. (5.31) in \cite{Spe14}), and defining
$\beta=r/n$, we get
\begin{align*}
k^{-n}u_{k,r}(n)\leq&
    k^{-n} (k-1)^{r} \sum_{d=0}^{\lfloor r/2 \rfloor}{
    \binom{n}{d}^2
    \binom{n-d+1}{r-2d}
    \left(\frac{k}{(k-1)^2}\right)^d}\\
\leq&
    (n+1)\sum_{d=0}^{\lfloor r/2 \rfloor}
        {2^{2n H\left(\frac{d}{n}\right) + 
    (n-d)H\left(\frac{r-2d}{n-d}\right) +
    (r-2d)\log_2{(k-1)} +
    (d-n)\log_2{k}}} \\
=&  (n+1)\sum_{d=0}^{\lfloor \beta n/2 \rfloor}{
    2^{ng_k\left(\beta,\frac{d}{n}\right)}}.
\end{align*}
\end{proof}

\begin{theorem} \label{thm:bound:eval}
For integer $k \geq 2$ and real $\beta \in [0,1]$, define the real function
\begin{align}\label{eq:G:g}
G_k(\beta) = \max_{0\leq \delta \leq \beta/2}{g_k(\beta,\delta)}
\end{align}
and the set of real numbers
\begin{align}\label{eq:A:G}
A_k = \{\beta \in [0,1] : G_k(\beta)<0\}.
  \end{align}
Then,
\begin{align*}
  \alpha_k \geq \beta_k^{*} := \sup A_k .
\end{align*}
\end{theorem}
\begin{proof}
First, we observe that the definition of $G_k(\beta)$ is well posed;
in fact, for any fixed $\beta\in[0,1]$, the function $g_k(\beta,
\delta)$ is bounded and continuous with respect to $\delta$, hence it
attains a maximum value in the compact set $0\leq \delta \leq \beta/2$
(by Weierstrass Theorem).

Second, we observe that $A_k$ is not empty, since $G_k(0)<0$.  In
fact, when $\beta=0$, the condition $\delta \in [0,\beta/2]$ is
satisfied only by $\delta=0$, and $g_k(0,0){=}-\log_2{k}< 0$, for any
$k\geq2$. Finally, since $A_k \subseteq [0,1]$, then $\sup{A_k}\leq
1$.

For $\beta \in A_k$, letting $f(n) = (n+1)\left(\left\lfloor
\frac{\beta n}{2} \right\rfloor +1\right)$, we see from Lemma
\ref{thm:exponential} that
\begin{align*}
k^{-n}u_{k,\beta n}(n)
    &\leq (n+1)\sum_{d=0}^{\lfloor \beta n/2 \rfloor}
    {2^{ng_k\left(\beta,\frac{d}{n}\right)}} 
    \leq f(n)2^{nG_k(\beta)},
\end{align*}
where we have used the relation $g_k\left(\beta,\frac{d}{n}\right)
\leq G_k(\beta)$. The latter follows from the definition of
$G_k(\beta)$ and the fact that, in each of the $\lfloor \frac{\beta
  n}{2} \rfloor +1$ terms of the summation, $0 \leq \frac{d}{n} \leq
\beta/2$. Taking now the limit in (\ref{eq:alpha:lb}) with $r^* =\beta
n$ yields:

\begin{align*}
\alpha_k \geq \lim_{n \rightarrow \infty}{\beta
    \left(1-f(n)2^{nG_k(\beta)}\right)} = \beta ,
\end{align*}
as $f(n)=\bigO(n^2)$ and $2^{nG_k(\beta)}$ is a negative exponential.
In conclusion, since $\alpha_k$ is no smaller than any member of
$A_k$, it is also no smaller than $\beta_k^*=\sup{A_k}$.
\end{proof}

As a first application of Theorem \ref{thm:bound:eval}, we obtain an
analytical lower bound to each $\alpha_k$. This bound is generally not
the best that can be obtained numerically from the theorem, but does
provide some insight. In particular, it shows that, as $k$ grows,
both $\beta_k^*$ and $\alpha_k$ approach~$1$.
\begin{proposition}\label{prop:growing-k}
Let the constant $M$ be defined as
\begin{align*}
      M = \max_{0{\leq}\beta{\leq}1,0{\leq}\delta{\leq}\beta/2}
      {2H(\delta) +
        (1-\delta)H\left(\frac{\beta-2\delta}{1-\delta}\right)}
      \approx 2.52.
\end{align*}
Then, for any $k \geq 3$, we have
\begin{align*}
     \alpha_k \geq \hat{\beta}_k = 1-\frac{M}{\log_2 (k-1)}.
\end{align*}
Two obvious corollaries are that $\lim_{k \rightarrow \infty}
\beta_k^* = 1$ and $\lim_{k \rightarrow \infty} \alpha_k = 1$.
\end{proposition}
\begin{proof}
For $k < 7$, $\hat{\beta}_k < 0$, thus, $\alpha_k \geq \hat{\beta}_k$
is trivially satisfied (recall that $\alpha_k \geq 0$). Hence, we
assume, for the remainder of the proof, $k \geq 7$ so that
$\hat{\beta}_k \geq 0$, leading to a well posed definition of $M$,
since it involves a maximum of a bounded, continuous function over a
compact domain.  We need to show that $g_k(\hat{\beta}_k,\delta) < 0$
for any $\delta\in[0,\hat{\beta}_k/2]$. By plugging the definition of
$\hat{\beta}_k$ in (\ref{eq:g:exponent}), after simple manipulations,
we obtain
\begin{align*}
  g_k(\hat{\beta}_k,\delta) =& -\left[\log_2 k -
    \log_2(k-1)\right]-\left[\delta(2\log_2(k-1)-\log_2 k)\right]\\
  &-\left[M- \left(2H(\delta) +
    (1-\delta)H\left(\frac{\hat{\beta}_k-2\delta}{1-\delta}\right)\right)\right]<0.
\end{align*}
It is straightforward to check that the terms within each of
the first two pairs of square brackets are positive for every
$k \geq 3$, while the expression within the third pair of square
brackets is non-negative (by the definition of $M$).
Finally, we clearly have $\hat{\beta}_k \leq \beta_k^* \leq \alpha_k
<1$, hence the stated limits are implied by the fact that
$\lim_{k \rightarrow \infty} \hat{\beta}_k = 1$.
\end{proof}

The corollary $\lim_{k \rightarrow \infty} \alpha_k=1$ also follows
from the result $\lim_{k \rightarrow \infty} \gamma_k=0$ (Theorem 1
in \cite{CS75}), together with the relationship
$1-\gamma_k\leq\alpha_k$, already mentioned in the introduction.

\subsection{Numerical computation of the lower bound}
In this subsection, we develop numerical
procedures, both to decide whether a specific (rational) number
$\beta$ qualifies as a lower bound to $\alpha_k$, according to Theorem
\ref{thm:bound:eval} (\emph{i.e.}, whether $g_k(\beta,\delta)<0$ for
every $\delta\in[0,\beta/2]$ or, equivalently, $G_k(\beta)<0$) and to
obtain the lower bound that subsumes all the $\beta$'s (that is,
$\beta_k^*=\sup{A_k}$).  The procedures presented below are based on
some properties of the functions $g_k(\beta,\delta)$ and $G_k(\beta)$,
which will be established, along the following lines.
\begin{itemize}
    \item We show
analytically that $g_k(\beta,\delta)$, when viewed as a function of
$\delta$, for some fixed value of $\beta$, achieves its maximum at a
unique point in its domain. By a bisection-like procedure, driven by
the sign of the derivative $\frac{\partial g_k}{\partial \delta}$,
lower and upper bounds to such maximum, $G_k(\beta)$, can be computed
with any desired accuracy.
\item We then develop a partial procedure
that returns $\sign(G_k(\beta))$, when $G_k(\beta)\not=0$, and does
not halt otherwise.
\item Finally, we show, analytically, that
$G_k(\beta)$ is an increasing function taking both negative and
positive values, so that the point $\beta_k^*=\sup A_k$ is the only
root of the equation $G_k(\beta)=0$ and can be (arbitrarily)
approximated by a bisection-like procedure, driven by the sign of
$G_k(\beta)$.
\end{itemize}

\subsubsection{Computing $G_k(\beta)$}
For simplicity, in this subsection, we adopt an idealized
\emph{infinite precision model}, where we assume that (i) the (real)
numbers arising throughout the computation are represented with
infinite precision and (ii) the results of the four basic
arithmetic operations, of comparisons, and of logarithms are computed
exactly.  We will discuss how to deal with the somewhat subtle issues
of finite precision in the next subsection.

We begin by identifying an interval that contains $G_k(\beta)$ and then
show how this interval can be made arbitrarily small.
\begin{proposition}\label{prop:G-beta-bound}
  Let $0 \leq \beta \leq 1$ and let $0<\delta^l<\delta^r<\beta/2$ be
  such that $0< \frac{\partial g_k}{\partial \delta}(\beta,\delta^l) <
  \infty$ and $-\infty < \frac{\partial g_k}{\partial
    \delta}(\beta,\delta^r) <0$. Let
\begin{align*}
  y(\delta^l,\delta^r)=\frac{\frac{\partial g_k}{\partial
      \delta}(\beta,\delta^l) \frac{\partial g_k}{\partial
      \delta}(\beta,\delta^r) (\delta^l -\delta^r)+ \frac{\partial
      g_k}{\partial \delta}(\beta,\delta^l) g_k(\beta,\delta^r)-
    \frac{\partial g_k}{\partial \delta}(\beta,\delta^r)
    g_k(\beta,\delta^l)}
  {\frac{\partial g_k}{\partial \delta}(\beta,\delta^l)-
\frac{\partial g_k}{\partial \delta}(\beta,\delta^r)}
\end{align*}
be the ordinate of the intersection of the two straight lines tangent
to the curve $g_k(\beta,\delta)$ (for fixed $\beta$) at $(\delta^l,
g_k(\beta,\delta^l))$ and $(\delta^r, g_k(\beta,\delta^r))$,
respectively. Then
\begin{align}\label{eq:G-interval}
G_k(\beta) \in [\max\left(g_k(\beta,\delta^l),
  g_k(\beta,\delta^r)\right), y(\delta^l,\delta^r)].
\end{align}
\end{proposition}
To prove the above proposition, we will need the following lemma,
which highlights some useful properties of $\frac{\partial
  g_k}{\partial \delta}$.
\begin{lemma}\label{lemma:delta-derivatives-g}
Let $0 \leq \beta \leq 1$. Then, as $\delta$ increases from $0$ to
$\beta/2$, the derivative $\frac{\partial g_k}{\partial \delta}$
decreases from $+\infty$ to $-\infty$ and vanishes at a unique point,
$\zeta_k(\beta)$, where $g_k(\beta,\zeta_k(\beta))=G_k(\beta)$.
\end{lemma}
\begin{proof}
From Equation~(\ref{eq:g:exponent}), basic calculus operations yield
\begin{align} \label{eq:Dg:Ddelta}
    \frac{\partial g_k}{\partial \delta}
    =& \log_2 \frac{k}{(k-1)^2}+2H'(\delta)-
      H\left(\frac{\beta-2\delta}{1-\delta}\right)-\frac{2-\beta}{1-\delta}
      H'\left(\frac{\beta-2\delta}{1-\delta}\right).
\end{align}
Considering that $\lim_{x \rightarrow 0^{+}} H'(x)=+\infty$ (see
Definition~\ref{def:H(x)}), we have that $\lim_{\delta
\rightarrow 0^{+}}\frac{\partial g_k}{\partial \delta}=+\infty$ (due
to the second term), while $\lim_{\delta \rightarrow
\beta/2^{-}}\frac{\partial g_k}{\partial \delta}=-\infty$ (due to
the fourth term).  Taking one more derivative, after some cancellation
of terms and simple rearrangements, we have
\begin{align} \label{eq:D2g:D2delta}
    \frac{\partial^2 g_k}{\partial^2 \delta}
     = 2 H''(\delta)+\frac{(2-\beta)^2}{(1-\delta)^3}
      H''\left(\frac{\beta-2\delta}{1-\delta}\right) <0,
\end{align}
where the last inequality follows from the fact that $H''(x)<0$, for
any $0\leq x \leq 1$ (see Definition~\ref{def:H(x)}). From
$\frac{\partial^2 g_k}{\partial^2 \delta}<0$, we have that
$\frac{\partial g_k}{\partial \delta}$ is strictly decreasing and,
considering that $\frac{\partial g_k}{\partial
\delta}(\beta,0^+)=+\infty$ and $\frac{\partial g_k}{\partial
\delta}(\beta, \beta/2^-)=-\infty$, we conclude that
$\frac{\partial g_k}{\partial \delta} $ takes each real value exactly
once.  Let then $\zeta_k(\beta)$ be the (unique) point where
$\frac{\partial g_k}{\partial \delta}(\beta,\zeta_k(\beta))=0$. It is
straightforward to argue that this is the unique point of maximum of
$g_k(\beta,\delta)$ (with respect to $\delta$, for fixed
$\beta$). Then, according to the definition of $G_k$ (see
Theorem~\ref{thm:bound:eval}), $G_k(\beta)=g_k(\beta,\zeta_k(\beta))$.
\end{proof}

\begin{proof}
(of Proposition~\ref{prop:G-beta-bound})
To better follow this proof, the reader may refer to the graphical illustration provided in Figure~\ref{fig:bisection}.
The lower bound to
$G_k(\beta)$, in Equation~(\ref{eq:G-interval}), trivially follows from
the fact that $G_k(\beta)$ is the maximum value of
$g_k(\beta,\delta)$, for $0 \leq \delta \leq \beta/2$. To establish
the upper bound, let $x(\delta^l,\delta^r)$ be the abscissa of the
intersection of the tangents (gray lines in 
Figure~\ref{fig:bisection}) considered in the statement (these
tangents do intersect, since they have different slopes). Due to the
downward convexity ($\frac{\partial^2 g_k}{\partial^2 \delta}<0$) of
$g_k(\beta, \delta)$, for any $\delta \in
[\delta^l,x(\delta^l,\delta^r)]$, the graph of $g_k(\beta,\delta)$
lies below the tangent at $(\beta,\delta^l)$.  Symmetrically, for
any $\delta \in [x(\delta^l,\delta^r),\delta^r]$, the graph of
$g_k(\beta,\delta)$ lies below the tangent at
$(\beta,\delta^r)$. Hence, for any $\delta \in [\delta^l,\delta^r]$,
the graph of $g_k(\beta,\delta)$ lies below the ordinate
$y(\delta^l,\delta^r)$ of the intersection of the two tangents.  In
particular, $G_k(\beta)=g_k(\beta,\zeta_k(\beta)) \leq
y(\delta^l,\delta^r)$, since $\zeta_k(\beta) \in
[\delta^l,\delta^r]$.
\end{proof}

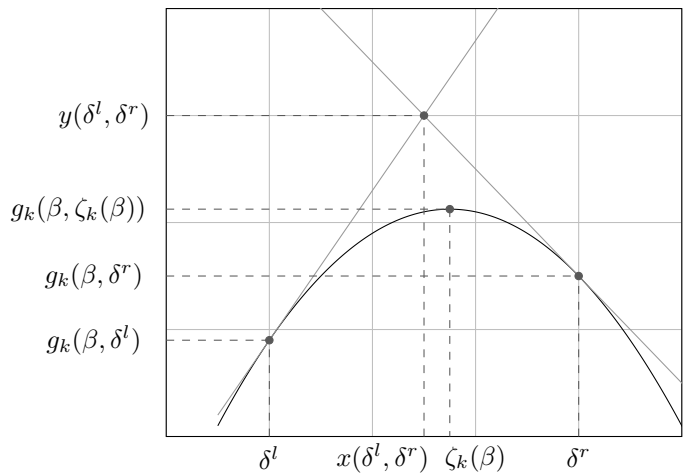
\begin{figure}[tbp]
\centering
\begin{tikzpicture}
\tikzset{every mark/.append style={scale=0.7}}
\begin{axis}[
    ticks=none, grid=major,
    ymin=-20, ymax=20, xmin=-6, xmax=4,
    clip mode=individual
]
\addplot[color=black, samples=200]{-x^2 -x +1};
\addplot[color=black!40, samples=50]{7*x+17};
\addplot[color=black!40, samples=50]{-5*x+5};

\draw[color=black!65,dashed] (axis cs: -6,-11) -- (axis cs:-4,-11);
\node at (axis cs:-7.4,-11) {$g_k(\beta, \delta^l)$};
\draw[color=black!65,dashed] (axis cs: -4,-20) -- (axis cs:-4,-11);
\node at (axis cs:-4,-22) {$\delta^l$};

\draw[color=black!65,dashed] (axis cs: -6, 1.25) -- (axis cs: -0.5, 1.25);
\node at (axis cs:-7.7,1.25) {$g_k(\beta, \zeta_k(\beta))$};
\draw[color=black!65,dashed] (axis cs: -0.5, 1.25) -- (axis cs: -0.5, -20);
\node at (axis cs: 0, -22) {$\zeta_k(\beta)$};

\draw[color=black!65,dashed] (axis cs:-6,-5) -- (axis cs:2,-5);
\node at (axis cs:-7.4,-5) {$g_k(\beta, \delta^r)$};
\draw[color=black!65,dashed] (axis cs: 2,-20) -- (axis cs:2,-5);
\node at (axis cs:2,-22) {$\delta^r$};

\draw[color=black!65,dashed] (axis cs:-6,10) -- (axis cs:-1,10);
\node at (axis cs:-7.2,10) {$y(\delta^l,\delta^r)$};
\draw[color=black!65,dashed] (axis cs:-1,-20) -- (axis cs:-1,10);
\node at (axis cs: -1.8, -22) {$x(\delta^l,\delta^r)$};

\addplot[color=black!65, only marks, mark=*] coordinates {
    (-4,-11)
    (-0.5, 1.25)
    (2,-5)
    (-1, 10)
};

\end{axis}
\end{tikzpicture}
\caption{Illustration of the proof of Proposition~\ref{prop:G-beta-bound}.}
\label{fig:bisection}
\end{figure}

\begin{proposition}[Computability of $G_k$, in the infinite precision
    model]  \label{prop:G:bisection}
There is a procedure (described in the proof) which,
  given as inputs an integer $k \geq 2$ and two real values $\beta \in
  [0,1]$ and $\epsilon>0$, outputs an interval $[G',G'']$ such that
  $G_k(\beta) \in [G',G'']$ and $G''-G'<\epsilon$.
  \end{proposition}
\begin{proof}
 The proposed procedure distinguishes 2 cases.

  \emph{Case 1}: $\frac{\partial g_k}{\partial \delta}(\beta,\beta/4)
  = 0$. Here, $\beta/4= \argmax_{\delta} g_k(\beta,\delta)$, whence
  $G_k(\beta)=g_k(\beta,\beta/4)$. The procedure outputs
  $G'=G''=g_k(\beta,\beta/4)$, clearly satisfying the requirements in
  the statement.
  
\emph{Case 2}: $\frac{\partial g_k}{\partial \delta}(\beta,\beta/4)
\not= 0$. Here the procedure includes two phases.  In a first phase, a
bisection process determines two points, $\delta_0^l$ and
$\delta_0^r$, which satisfy the assumptions of
Proposition~\ref{prop:G-beta-bound}. (For simplicity, the dependence
of $\delta_0^l$ and $\delta_0^r$ upon $k$ and $\beta$ is not made
explicit in the notation.)  In a second phase, the interval
$[\delta_0^l, \delta_0^r]$ is iteratively bisected, until the interval
appearing in Equation~(\ref{eq:G-interval}) has size smaller than
$\epsilon$. In both phases, the bisection is driven by the sign of
$\frac{\partial g_k}{\partial \delta}$.  The first phase includes two
subcases:

\emph{Subcase 2a}: $\frac{\partial g_k}{\partial
  \delta}(\beta,\beta/4) >0$. We define the sequence
$\mu_j=(1-2^{-j})(\beta/2)$, for $j \geq 1$.  Letting $h=\min\{j
\geq 2: \frac{\partial g_k}{\partial \delta}(\beta, \mu_j)<0\}$, we
set $[\delta_0^l,\delta_0^r]=[\mu_{h-1},\mu_h]$.

\emph{Subcase 2b}: $\frac{\partial g_k}{\partial
  \delta}(\beta,\beta/4)<0$.  We define the sequence
$\nu_j=2^{-j}(\beta/2)$, for $j \geq 1$.  Letting $h=\min\{j \geq
2: \frac{\partial g_k}{\partial \delta}(\beta, \nu_j)>0\}$, we set
$[\delta_0^l,\delta_0^r]=[\nu_{h},\nu_{h-1}]$.

Since $\lim_{j \rightarrow \infty} \mu_j = \beta/2$ and $\lim_{j
  \rightarrow \infty} \nu_j = 0$, the interval
$[\delta_0^l,\delta_0^r]$ is well defined, in either subcase. Its
endpoints can be computed by iteratively testing the condition on the
derivative for $j=1,2,\ldots$, till it is satisfied.

In both subcases, $\zeta_k(\beta) \in [\delta_0^l,\delta_0^r]$. We
can then construct a sequence of intervals, each half the size of
the preceding one and containing $\zeta_k(\beta)$, as follows.

For $i=1,2, \ldots$ do:
\begin{enumerate}
  \item $c_i=(\delta_{i-1}^l+\delta_{i-1}^r)/2$.
  \item If $\frac{\partial g_k}{\partial \delta}(\beta, c_i) = 0$,
    then set $G'=G''=g_k(\beta,c_i)$ and exit.
  \item If $\frac{\partial g_k}{\partial \delta}(\beta, c_i) > 0$,
    then let $[\delta_i^l,\delta_i^r]=[c_i,\delta_{i-1}^r]$.
  \item If $\frac{\partial g_k}{\partial \delta}(\beta, c_i) < 0$,
    then let $[\delta_i^l,\delta_i^r]=[\delta_{i-1}^l,c_i]$.
  \item Set $G'=\max(g_k(\beta,\delta_i^l), g_k(\beta,\delta_i^r))$
    and $G''= y(\delta_i^l,\delta_i^r)$, as defined in
    Proposition~\ref{prop:G-beta-bound}. If $G''-G' < \epsilon$, then
    exit.
\end{enumerate}
It is straightforward to show that, for any $\eta>0$, $\log_2
(1/\eta)$ bisection iterations (counting those of both phases) are
sufficient to guarantee $\delta^r-\delta^l \leq \eta$, hence to
determine $\zeta_k(\beta)$ with accuracy $\eta>0$.

It remains to show that the above for loop is eventually exited. If
the loop is exited at step 2, then we are are done.  Otherwise, as we
will argue, $G''-G'$ vanishes with $\delta^r-\delta^l$ so that, for
some $i$, $G''-G'< \epsilon$, hence the loop is eventually exited, at
step 5. Toward this conclusion, we can observe that
\begin{align*}
  y(\delta^l,\delta^r)-g_k(\beta,\delta^l) & \leq \left|
  \frac{\partial g_k}{\partial \delta}(\beta,\delta^l) \right|
  (x(\delta^l,\delta^r)-\delta^l),\\ y(\delta^l,\delta^r) -
  g_k(\beta,\delta^r) & \leq \left| \frac{\partial g_k}{\partial
    \delta}(\beta,\delta^r) \right| (\delta^r - x(\delta^l,\delta^r)),
\end{align*}
where $x(\delta^l,\delta^r)$ is the abscissa of the intersection of
the tangents. As $\delta^r-\delta^l$ approaches $0$, we have that
\begin{align*}
& \lim \delta^l = \lim \delta^r = \lim x(\delta^l,\delta^r) =
  \zeta_k(\beta),\\ & \lim \frac{\partial g_k}{\partial
    \delta}(\beta,\delta^l) = \lim \frac{\partial g_k}{\partial
    \delta}(\beta,\delta^r) =0,\\ & \lim G''-G' = \lim
  y(\delta^l,\delta^r) -
  \max(g_k(\beta,\delta^l),g_k(\beta,\delta^r))=0.\\[-10mm]
\end{align*}
\end{proof}

\subsubsection{Computing $\sign(G_k(\beta))$}
\label{sec:sign:G}
 
We now consider the more realistic model where only rational numbers
are represented, and only arithmetic operations and comparisons with
rational inputs can be computed exactly. If $f(x)$ is a real function
of the real variable $x$, we will restrict our attention to rational
inputs and be interested in computing arbitrary approximations of
$f(x)$.  More specifically, our target is an \emph{approximation
algorithm} ${\cal A}_f(x,\epsilon)$ of the rational inputs $x$ and
$\epsilon>0$, whose output $\tilde{f}(x,\epsilon)$ satisfies the
relationship $|\tilde{f}(x,\epsilon)-f(x)| < \epsilon$. Thus, $f(x)$
is a computable real number, in the sense of
Definition~\ref{def:computability}. These notions naturally extend to
functions of several variables.

Approximation algorithms are well known for $\log_2 x$, such as the
iterative method of \cite{ML73}.  For evaluating a
finite arithmetic expression, a tedious but straightforward analysis
of error propagation will determine the required accuracy of each
intermediate calculation capable of guaranteeing the desired accuracy
for the target result. Therefore, approximation algorithms can be
derived for $g_k(\beta,\delta)$ and $\frac{\partial g_k}{\partial
\delta}(\beta,\delta)$.

One issue we need to deal with is that an algorithm to compute $f(x)$
with any desired accuracy does not automatically translate into an
algorithm to systematically determine the sign (positive, negative, or
zero), of $f(x)$, denoted $\sign(f(x))$.  In fact, while
$\tilde{f}(x,\epsilon) \geq +\epsilon$ implies $f(x)>0$ and
$\tilde{f}(x,\epsilon) \leq -\epsilon$ implies $f(x)<0$, in the
remaining case, $|\tilde{f}(x,\epsilon)| < \epsilon$, nothing can be
inferred about $\sign(f(x))$.  On the one hand, when $f(x)=0$, this
``undeterminate'' case is bound to occur, for any $\epsilon>0$.  On
the other hand, when $f(x) \not=0$, this case will not occur if
$\epsilon \leq |f(x)|/2$, since $|\tilde{f}(x, \epsilon)-f(x)| <
\epsilon \leq |f(x)|/2$ implies $|\tilde{f}(x, \epsilon)| \geq
|f(x)|/2 \geq \epsilon$.

The preceding observations suggest the following procedure
$\mathcal{S}(\mathcal{A}_f,x)$ which, building on an algorithm
$\mathcal{A}_f$ that computes an approximation $\tilde{f}(x,\epsilon)$
to $f(x)$, will halt and output $\sign(f(x)) \in \{-,+\}$, when
$f(x)\not=0$ and will not halt when $f(x)=0$. Given a monotonically
vanishing, computable sequence $\epsilon_1,\epsilon_2, \ldots$
(\emph{e.g.}, $\epsilon_i=2^{-i}$), for $i=1,2,\ldots$, the procedure
$\mathcal{S}$ computes $\tilde{f}_i = \tilde{f}(x,\epsilon_i)$ by a
call to $\mathcal{A}_f(x,\epsilon_i)$; if $\tilde{f}_i \geq
\epsilon_i$ or $\tilde{f}_i \leq -\epsilon_i$, then
it returns $\sign(\tilde{f}_i)$ and halts.

The bisection procedure has to be modified so that it will not get
stuck in the attempt of evaluating the sign of a zero. Assuming that
$f$ has a unique zero in the interval $[a,b]$, say, with $f(a)<0$ and
$f(b)>0$, a point $c$ where to split the interval can be found as
follows.  Let $c'=(2a+b)/3$ and $c''=(a+2b)/3$ the points that
\emph{trisect} the interval.  Interleave the executions of the calls
$\mathcal{S}(\mathcal{A}_f,c')$ and $\mathcal{S}(\mathcal{A}_f,c'')$
until the termination of one of them, an event guaranteed to occur since
at least one between $f(c')$ and $f(c'')$ differs form zero.  Let $c$
be the argument for which the execution has terminated, whence
  $\sign(f(c))$ has been determined. The refined interval is chosen to
  be $[a,c]$ when $f(c)>0$ and $[c,b]$ when $f(c)<0$.  In all cases, the
interval size shrinks at least by $2/3$.

With the tools we have introduced, we can now tackle the
finite-precision computation of $\sign(G_k(\beta))$, a quantity that
will play a key role in the computation of $\beta_k^*$ discussed in the next subsection.

\begin{proposition}[Partial computability of $\sign(G_k(\beta))$, in the
finite precision model]  \label{prop:sign:G} There is a
procedure (described in the proof) which, given as inputs an integer
$k \geq 2$ and two rational values $\beta \in [0,1]$ and $\epsilon>0$,
outputs $\sign(G_k(\beta))$ and halts, if $G_k(\beta) \not=0$,
and does not halt, if $G_k(\beta)=0$.
\end{proposition}
\begin{proof}
We obtain the procedure for $\sign(G_k(\beta))$ by adapting the
procedure to approximate $G_k(\beta)$ presented in the proof of
Proposition~\ref{prop:G:bisection}.  To lighten the notation,
throughout this proof, we let $g(\delta)=g_k(\beta,\delta)$ and
$g'(\delta)=\frac{\partial g_k}{\partial \delta}(\beta,\delta)$.

Interleave the executions of the calls
$\mathcal{S}(\mathcal{A}_{g'},\beta/4)$ and
$\mathcal{S}(\mathcal{A}_{g'},\beta/8)$ until termination of one of
them, an event guaranteed to occur since at least one between
$g'(\beta/4)$ and $g'(\beta/8)$ differs form zero.  Let $\delta_0$ be
the argument for which the execution has terminated. If
$g'(\delta_0)>0$, then search for a $\mu_h$ such that $g'(\mu_h)<0$
and let $[\delta_0^l,\delta_0^r]=[\delta_0,\mu_h]$.  Else
$g'(\delta_0)<0$, search for a $\nu_h$ such that $g'(\nu_h)>0$ and let
$[\delta_0^l,\delta_0^r]=[\nu_h,\delta_0]$. The search has to
interleave the evaluation of two consecutive points in the sequence,
to avoid the potential non termination of the execution, which can
occur at most at one point.

The sequence $[\delta_i^l,\delta_i^r]$, for $i=1,2,\ldots$, of
successive refinements of $[\delta_0^l,\delta_0^r]$ is then
constructed, choosing the splitting point $c_i$ by trisection.  Let
$\epsilon_i$ be the error bound such that the call to
$\mathcal{A}_{g'}(c_i,\epsilon_i)$ has enabled the determination of
the sign of $g'(c_i)$. Let also $\eta_i=\min(2^{-i},\epsilon_i)$ and
compute $\eta_i$-approximations $\tilde{g}(\delta_i^l,\eta_i)$,
$\tilde{g}(\delta_i^r,\eta_i)$, and
$\tilde{y}(\delta_i^l,\delta_i^r,\eta_i)$. The conditions that
enable determining $\sign(G_k(\beta))$ are as follows:
\begin{itemize}
\item If $\tilde{g}(\delta_i^l,\eta_i) \geq \eta_i$ or
 $\tilde{g}(\delta_i^r,\eta_i) \geq \eta_i$, then return
 $\sign(G_k(\beta))=+$ and halt.
\item If $\tilde{y}(\delta_i^l,\delta_i^r,\eta_i) \leq -\eta_i$,
 then return $\sign(G_k(\beta))=-$ and halt.
\end{itemize}
Indeed, in the first case, the condition implies that
$g(\delta_i^l)=g_k(\beta,\delta_i^l)>0$ or
$g(\delta_i^r)=g_k(\beta,\delta_i^r)>0$, which in turn implies
$G_k(\beta)>0$. Conversely, if $G_k(\beta)>0$, for $i$ large enough,
both $g(\delta_i^l)$ and $g_k(\delta_i^r)$, which are monotonically
non decreasing with $i$, become positive. Since $\eta_i$ vanishes
with $i$, it will eventually become sufficiently small to satisfy
the condition.

Symmetrically, in the second case, the condition implies that
$G_k(\beta)<0$. Conversely, if $G_k(\beta)<0$, for $i$ large enough,
$y(\delta_i^l,\delta_i^r)$, which decreases with $i$, becomes
negative. Again, $\eta_i$ will eventually become sufficiently small
to satisfy the condition.

Finally, we observe that if $G_k(\beta)=0$, then neither condition
will ever hold and the procedure will not halt.
\end{proof}

\subsubsection{Computing $\beta^*_k=\sup A_k$}\label{sec:decide}

In this subsection, we will see that the set $A_k$ is an
interval, closed on the left and open on the right: specifically,
$A_k=[0,\beta_k^*)$. On the one hand, this property of $A_k$ is not
surprising, since if $\beta'$ is a lower bound to $\alpha_k$ and
$\beta<\beta'$, then $\beta$ is a lower bound too. 
On the other
hand, the property does require a proof, since membership in $A_k$
is a sufficient, but not necessary condition for $\beta$ to be a
lower bound, a scenario compatible with $A_k$ having ``holes'',
that is, with $G_k(\beta)$ taking both negative and positive
values to the left of $\beta_k^*$.  However, we will argue that
$A_k$ has no such holes, since $G_k(\beta)$ has a unique
zero in $[0,1]$ which, by Equation~(\ref{eq:A:G}), is our target,
$\beta^*_k=\sup A_k$. We will show this zero to be computable.
\begin{proposition}\label{prop:G-monotone}
For every $k \geq 2$, in the interval $0 \leq \beta \leq 1$, the
function $G_k(\beta)$ is increasing and has a unique zero,
which equals $\beta_k^* = \sup A_k$.
\end{proposition}
\begin{proof}
From Lemma~\ref{lemma:delta-derivatives-g}, we have that
$G_k(\beta) = g_k(\beta,\zeta_k(\beta))$, where $\zeta_k(\beta)$ is
the unique solution of the equation (in $\delta$)
  $\frac{\partial g_k}{\partial \delta}(\beta,\delta)=0$, so that
\begin{align}\label{eq:delta-derivative-g-vanishes}
\frac{\partial g_k}{\partial \delta}(\beta,\zeta_k(\beta))=0.
\end{align}
By the chain rule for the total derivative and the above
relationship, we get
\begin{align*}
\frac{\partial G_k}{\partial \beta} &=
\frac{\partial g_k}{\partial \beta}(\beta,\zeta_k(\beta))
\frac{\partial \beta}{\partial \beta}
+ \frac{\partial g_k}{\partial \delta}(\beta,\zeta_k(\beta))
\frac{\partial \zeta_k(\beta)}{\partial \beta}
= \frac{\partial g_k}{\partial \beta}(\beta,\zeta_k(\beta))\\
&=\log_2{(k-1)} + H'\left(\frac{\beta-2\zeta_k(\beta)}{1-\zeta_k(\beta)}\right),
\end{align*}
where the last step follows from Equation~(\ref{eq:g:exponent}).
After several but simple manipulations and letting
$\zeta=\zeta_k(\beta)$, the expression just derived
together with
Equations~(\ref{eq:delta-derivative-g-vanishes})
and~(\ref{eq:Dg:Ddelta}) yield the relationship
\begin{align}\label{eq:signG-analysis0}
\left(\frac{2-\beta}{1-\zeta}\right)\frac{\partial G_k}{\partial
  \beta} = \log_2{\frac{k}{k-1}} + 2H'(\zeta) -
H\left(\frac{\beta-2\zeta}{1-\zeta}\right) +
\left(1-\frac{\beta-2\zeta}{1-\zeta}\right)\log_2{(k-1)}.
\end{align}
We will argue that the right hand side is positive, for any $\beta \in
[0,1]$ and any $\zeta \in [0,\beta/2]$.  As easily seen, the
only negative term in the right hand side of
Equation~(\ref{eq:signG-analysis0}) is
$-H\left(\frac{\beta-2\zeta}{1-\zeta}\right) \geq -1$. A case
analysis shows that this negative term is offset by
one or more of the remaining terms, yielding a positive sum.

\emph{Case 1}: $\zeta<1/(1+\sqrt{2})$. Considering that $H'$ is
monotonically decreasing (see Definition~\ref{def:H(x)}), we have
\begin{align*}
  2H'(\zeta) > 2H'\left(\frac{1}{1+\sqrt{2}}\right)=2 \log_2
  \frac{1-1/(1+\sqrt{2})}{1/(1+\sqrt{2})}
       = 2 \log_2 \sqrt{2}=1.
\end{align*}
$~~$
\emph{Case 2}: $\zeta \geq 1/(1+\sqrt{2})$ and $k \geq 4$. We observe
that $\left(1-\frac{\beta-2\zeta}{1-\zeta}\right)$ is decreasing with
$\beta$ and increasing with $\zeta$, whereas $\log_2(k-1)$ is
increasing with $k$, whence
\begin{align*}
  \left(1-\frac{\beta-2\zeta}{1-\zeta}\right)\log_2(k-1) >
  \left(1-\frac{1-2/(1+\sqrt{2})}{1-1/(1+\sqrt{2})}\right)\log_2(3)>
  \frac{1}{\sqrt{2}}1.58 > 1.11
\end{align*}

\emph{Case 3}: $\zeta \geq 1/(1+\sqrt{2})$ and $k=2$. The right hand
side of Equation~(\ref{eq:signG-analysis0}) becomes
\begin{align*}
 1 + 2H'(\zeta) - H\left(\frac{\beta-2\zeta}{1-\zeta}\right),
\end{align*}
which is always positive. In fact, if $\zeta<1/2$, then $2H'(\zeta)>0$
and $1- H\left(\frac{\beta-2\zeta}{1-\zeta}\right) \geq 0$. On the
other hand, if $\zeta=1/2$, then $2H'(\zeta)=0$; furthermore, it must
be $\beta=1$ so that
$H\left(\frac{\beta-2\zeta}{1-\zeta}\right)=H(0)=0$.

\emph{Case 4}: $\zeta \geq 1/(1+\sqrt{2})$ and $k=3$. Considering that
(i) $H'(\zeta) \geq 0$ (for $0 \leq \zeta \leq 1/2)$, (ii) $H(x) \leq
1$ (for any $0 \leq x \leq 1$), and (iii)
$\left(1-\frac{\beta-2\zeta}{1-\zeta}\right)>\frac{1}{\sqrt{2}}$ (as
seen in Case 2), Equation~(\ref{eq:signG-analysis0}) implies
\begin{align}\label{eq:signG-analysis4}
\left(\frac{2-\beta}{1-\zeta}\right)\frac{\partial G_k}{\partial
  \beta} \geq \log_2{\frac{3}{2}} - 1 +\frac{1}{\sqrt{2}} > 0.58 -1
+ 0.70 = 0.28 >0.
\end{align}
Having established that $G_k(\beta)$ has a positive
derivative, hence it is increasing, in the interval $[0,1]$, we
now argue the existence of a (unique) zero of $G_k(\beta)$, by
showing that $G_k(0)<0$ and $G_k(1)>0$, for every $k \geq 2$.

At $\beta=0$, we simply observe that, by definition
(\emph{i.e.}, Equation~(\ref{eq:G:g})), we have $G_k(0)=\max_{0 \leq \delta \leq
0} g_k(0,\delta)=g_k(0,0)=-\log_2 k <0$, for every $k \geq 2$.

At $\beta=1$, we observe that $G_k(1)=\max_{0 \leq \delta \leq
1/2} g_k(1,\delta) \geq g_k(1,1/k)$ and show that
$g_k(1,1/k)>0$, for every $k \geq 2$.  We consider the following
chain of relationships where, starting from
Equation~(\ref{eq:g:exponent}), (i) we have dropped the contribution
$(1-\delta)\log_2\frac{1}{1-\delta}$ to $H(\delta)$ and the second
entropy term, which are both non-negative; (ii) we have plugged
$\beta=1$ and $\delta=\frac{1}{k}$; and (iii) we have performed
simple algebraic manipulations, also making use of the inequality
$\log_2 (1+x)\leq (\log_2 e)x$:
\begin{eqnarray*}
  g_k(1,1/k) & \geq &
  \left(1-2\frac{1}{k}\right)\log_2(k-1)-\left(1-\frac{1}{k}\right)\log_2 k +
  2\frac{1}{k}\log_2 k\\
  &=& \log_2\frac{k-1}{k}
  +\frac{2}{k}\log_2\frac{k}{k-1}+\frac{1}{k}\log_2{k}\\
  &\geq& -(\log_2 e)\left(1-\frac{2}{k}\right)
  \frac{1}{k}+\frac{1}{k}\log_2{k}\\
  &=& \frac{1}{k}\left[\log_2{k}-(\log_2 e)\left(1-\frac{2}{k}\right)
  \right].
\end{eqnarray*}
For $k=2$, the expression within the square bracket evaluates to $1$.
For $k \geq 3$, it is easy to see that the square bracket is
positive, as $\log_2{k}-(\log_2 e)\left(1-\frac{2}{k}\right)>\log_2 3
- \log_2 e > 0$. In conclusion, for every $k \geq 2$, we have
$g_k(1,\frac{1}{k})>0$, as claimed.

Finally, given that $G_k(\beta)$ is increasing with $\beta$, its
unique zero is the supremum $\beta_k^*$ of $A_k$
(cf. Equation~(\ref{eq:A:G})).

\end{proof}

We can now present a simple numerical algorithm to approximate
$\beta_k^*$, from below, with any desired accuracy.
\begin{proposition}\label{prop:beta*::bisection}
(Computability of $\beta_k^*$, in the finite precision model.)
There is a procedure (outlined in the proof) which, on inputs $k
\geq 2$ and $\epsilon>0$, outputs a number $\bar{\beta}_k^* \in
[\beta_k^*-\epsilon,\beta_k^*]$.  Clearly, $\bar{\beta}_k^*
\leq \alpha_k$.
\end{proposition}

\begin{proof}
  In light of Proposition~\ref{prop:G-monotone}, it is straightforward
  to develop a trisection procedure that starts with the
  interval $[0,1]$, is driven by the sign of $G_k(\beta)$ computed
  according to Proposition~\ref{prop:sign:G}, stops after
  $\lceil \log_{3/2} \frac{1}{\epsilon} \rceil$ iterations, and outputs
  the left endpoint of the current interval containing $\beta_k^*$,
  which has size smaller than $(3/2)^{- \lceil \log_{3/2}
    \frac{1}{\epsilon}\rceil} \leq \epsilon$.
\end{proof}

\begin{table*}[t]
  \centering
    \caption{Lower bound $\bar{\beta}_k^{*}$ to $\alpha_k$, for
      various values of $k$. With reference to
      Proposition~\ref{prop:beta*::bisection}, $\epsilon = 10^{-8}$,
      hence the five digits to the right of the decimal point are
      guaranteed to be the same as those of $\beta_k^{*}$. The simple
      $1-1/k$ upper bound to $\alpha_k$ shows how, for large $k$,
      $\beta_k^*$, is quite close to $\alpha_k$.}
    \label{tab:lower:bounds:k}

    \vspace*{3mm}
    \begin{tabular}{ccc}
    \hline\\[-3mm]
    $k$  & $\bar{\beta}_k^{*}$ & $1-1/k$\\
    \\[-3mm]
    \hline
    \\[-3mm]
    $ 2$ & $0.17055$ & $0.50000$\\
    $ 3$ & $0.28366$ & $0.66667$\\
    $ 4$ & $0.35978$ & $0.75000$\\
    $ 5$ & $0.41517$ & $0.80000$\\
    $ 6$ & $0.45776$ & $0.83333$\\
    $ 7$ & $0.49183$ & $0.85714$\\
    $ 8$ & $0.51990$ & $0.87500$\\
    $16$ & $0.64475$ & $0.93750$\\
    $32$ & $0.73867$ & $0.96875$\\
    $2^{10}$ & $0.94359$ & $0.99902$\\
    $2^{20}$ & $0.99686$ & $0.99999$\\
    $2^{30}$ & $0.99978$ & $0.99999$\\
    $2^{40}$ & $0.99998$ & $0.99999$\\
    \hline
    \end{tabular}
\end{table*}
Table~\ref{tab:lower:bounds:k} reports the lower
bound $\bar{\beta}_k^*$ (approximating $\beta_k^*$ from below), for a set of values
of $k$.  The program we have used is based on a direct
implementation of the bisection version of the relevant
procedures. For the level of accuracy of the reported results,
the precision of standard floating point arithmetic turns out to
be sufficient. The execution is rather fast: even the result for
$k=2^{40}$ took about $15$ milliseconds of core-time to compute,
even using a non-optimized, straightforward implementation
of the bisection algorithm.
As a term of comparison, the table also shows the value of the
simple upper bound $\alpha_k \leq 1-\frac{1}{k}$ (from Hamming distance). Considering
that, from Proposition~\ref{prop:growing-k}, $\lim_{k \rightarrow
  \infty} (\alpha_k-\beta_k) = 0$, we see how, as $k$ increases,
$\beta_k^{*}$ provides an increasingly better approximation (from
below) to $\alpha_k$.  To guarantee the same level of approximation
via the lower bound $\alpha_k \geq \alpha_k(n)-Q(n)$, increasing
values of $n$ are required, as $k$ increases. Eventually, computing
$\beta_k^{*}$ becomes less expensive than estimating $\alpha_k(n)$.

\section{A conjecture on the asymptotic behavior of $\alpha_k$}
\label{sec:large:k:conjecture}

In the context of the LCS problem, \cite{KLM05} have proven the conjecture, proposed by \cite{SK83}, that $\lim_{k \rightarrow \infty}
  \gamma_k \sqrt{k} = 2$. As a corollary, we have that, for large $k$,
  $1-\alpha_k \geq \gamma_k \approx \frac{2}{\sqrt{k}}$.  Here, we
  propose, as a conjecture to be explored, that $\lim_{k \rightarrow
    \infty} (1-\alpha_k)k= c_{\alpha}$ for some constant $c_{\alpha}
  \geq 1$.  The bound $c_{\alpha} \geq 1$ follows from $\alpha_k \leq
  1-\frac{1}{k}$. We essentially conjecture that the limit is
  finite. Our intuition is that, although for large $n$ and $k$ we can
  expect an LCS of length approximately $\frac{2n}{\sqrt{k}}$, it
  would be too costly, in terms of insertions and deletions, to align
  more than $\bigO(\frac{n}{k})$ matches.

We tentatively investigate the conjecture numerically, by
  considering Monte Carlo estimates of the quantity
  $c_{\alpha,k}(n)=(1-\alpha_{k}(n))k$. Unfortunately, obtaining such
  estimates with the required precision presents some challenges.

Since we expect values of $c_{\alpha,k}(n)$ not much larger than
$1$, we need the error on $(1-\alpha_{k}(n))$, hence on
$\alpha_k(n)$, to be a fraction of $1/k$, a constraint that becomes
increasingly stringent as $k$ increases. In order to guarantee a
sufficient upper bound on the error, based on
Proposition~\ref{prop:aplha:k:n:confidence} and its corollary
Equation~(\ref{eq:alpha(n):confidence:radius}), the necessary values
of $n$ and $N$ quickly become prohibitive, as $k$
increases. However, we suspect that these errors bounds do become
rather loose, for large $k$.  More specifically, we hypothesize
that, for $x,y \in \Sigma_k^n$ with large $k$, the standard
deviation of $d_E(x,y)/n$ can be approximated as
$\sqrt{\frac{1}{kn}}$. This is based on the intuition that, for
large $k$, the edit distance behaves similarly to the Hamming
distance $d_H(x,y)$, whose standard deviation can be easily
determined to be $\sqrt{n\frac{1}{k}\left(1-\frac{1}{k}\right)}
\approx \sqrt{\frac{n}{k}}$.  We have tested this hypothesis by
experimentally estimating the standard deviation of $d_E(x,y)/n$,
for $n=2^{18}$ and $k=2^{14}$, based on $N=80$
independent pairs of random strings.  The estimated value turned
out to be $0.14~10^{-4}$, well in line with our assumption that
the standard deviation is approximately
$\sqrt{\frac{1}{kn}}=\sqrt{\frac{1}{2^{14}2^{18}}}= 2^{-16} \approx
0.15~10^{-4}$.

Table~\ref{tab:c:k:n} shows estimates of $c_{\alpha,k}(n)$ for
$n=2^{17}, 2^{18}, 2^{19}, 2^{20}$ and
$k=2^7,2^8,\ldots,2^{20}$. Highlighted in bold face are the
entries for which the hypothesized statistical error $k
\sqrt{\frac{1}{kn}}=\sqrt{\frac{k}{n}}$ on $c_{\alpha,k}(n)$ is at
most $\frac{1}{4}$, that is, $k \leq \frac{n}{16}$. We can observe that,
where the hypothesized error is small enough, $3 \leq c_{\alpha,k}(n)
\leq 4$.

Of course, how well $c_{\alpha,k}(n)$ approximates $c_{\alpha,k}$
depends on how well $\alpha_k(n)$ approximates $\alpha_k$.  The
quality of the latter approximation increases for large $k$, where the
the bound provided by Theorem~\ref{thm:computability:alpha} becomes
loose. In fact, putting together various results, we have that
\begin{equation}\label{eq:rate:convergence:k}
1 - \frac{M}{\log_2 (k-1)} \leq \alpha_k \leq \alpha_k(n) \leq 1 - \frac{1}{k}.
\end{equation}
We can see that the difference between the last and the first term,
hence the difference $\alpha_k(n)-\alpha_k$ between the intermediate
terms, vanishes when $k$ diverges.  The quantitative impact on the
difference $c_{\alpha,k}-c_{\alpha,k}(n)$ remains to be seen, but the
relative stability of the bold entries in each column of
Table~\ref{tab:c:k:n} seems compatible with assuming a small impact.
Of course, this section remains in the realm of conjectures, which
will hopefully provide some motivation for rigorous analysis that may
confirm or refute them.

 \begin{table*}[!ht]
   \centering
   \caption{Estimate of $c_{\alpha,k}(n)=(1-\alpha_k(n))k$ based
     on a single random pair ($N=1$). In bold face are shown the
     entries of the table where the (hypothesized) standard deviation
     of the error satisfies $\sqrt{k/n} \leq 1/4$.}
     \label{tab:c:k:n}
     \begin{tabular}{ccccc}
       \hline & \multicolumn{4}{c}{$n$}\\
       $k$ & $2^{17}$ & $2^{18}$ & $2^{19}$ & $2^{20}$ \\
       \hline
        \\[-3mm]
       $2^7$ & $\mathbf{3.553}$ & $\mathbf{3.552}$ & $\mathbf{3.566}$
       & $\mathbf{3.581}$ \\
       $2^8$ & $\mathbf{3.617}$ & $\mathbf{3.608}$ & $\mathbf{3.629}$
       & $\mathbf{3.635}$ \\
       $2^9$ & $\mathbf{3.531}$ & $\mathbf{3.656}$ & $\mathbf{3.658}$
       & $\mathbf{3.678}$ \\
       $2^{10}$ & $\mathbf{3.570}$ & $\mathbf{3.625}$ & $\mathbf{3.654}$
       & $\mathbf{3.669}$ \\
       $2^{11}$ & $\mathbf{3.406}$ & $\mathbf{3.516}$ & $\mathbf{3.668}$
       & $\mathbf{3.652}$ \\
       $2^{12}$ & $\mathbf{3.250}$ & $\mathbf{3.656}$ & $\mathbf{3.625}$
       & $\mathbf{3.672}$ \\
       $2^{13}$ & $\mathbf{3.000}$ & $\mathbf{3.500}$ &
       $\mathbf{3.500}$ & $\mathbf{3.703}$ \\
       $2^{14}$ & $3.125$ &  $\mathbf{3.250}$ & $\mathbf{3.563}$ &
       $\mathbf{3.672}$ \\
       $2^{15}$ & $3.000$ &
       $3.500$ & $\mathbf{3.750}$ & $\mathbf{3.313}$ \\
       $2^{16}$ & $3.500$ & $2.750$ & $3.250$ & $\mathbf{3.125}$ \\
       $2^{17}$ & $3.000$ & $2.000$ & $3.000$ & $3.000$ \\
       $2^{18}$ & $4.000$ & $4.000$ & $2.000$ & $2.500$ \\
       $2^{19}$ & $0.000$ & $2.000$ & $2.000$ & $1.500$ \\
       $2^{20}$ & $0.000$ & $0.000$ & $0.000$ & $0.000$\\
       \hline
     \end{tabular}
 \end{table*}

\section{Conclusions and further questions}
\label{sec:conclusion}
In this paper, we have explored ways to compute Monte Carlo
estimates, upper bounds, and lower bounds to the asymptotic constant
characterizing the expected edit distance between random, independent
strings. We have presented the theoretical basis for various
approaches and used them to obtain numerical results for some
alphabet sizes $k$, which improve over previously known values
\cite{GMR16}. However, there is still a significant gap between upper
and lower bounds that can be actually computed in a reasonable
time.  Below, we outline a number of open questions worthy
of further investigation.

The approaches proposed here can be extended to the study of other
statistical properties of the edit distance for a given length $n$,
\emph{e.g.}, the standard deviation, which has been widely studied in
the context of the longest common subsequence. Ultimately, a
characterization of the full distribution would be desirable.

The exact rate of convergence $q_k(n)=\alpha_k(n)-\alpha_k$, or even
just its asymptotic behavior, remains to be determined. In particular,
we are not aware of any significant lower bound to $q_k(n)$ to be
compared with the upper bound $Q(n)$. Moreover, this upper bound is 
oblivious to $k$, whereas the rate of convergence is affected by $k$, 
as indicated by Equation~(\ref{eq:rate:convergence:k}).

From Proposition~\ref{prop:growing-k} and a straightforward analysis
of ``Hamming scripts'', which use only matches and substitutions, we
know that $\frac{M}{\log_2 k} \leq \ 1-\alpha_k \leq \frac{1}{k}$,
(where $M \approx 2.52$); what is the exact asymptotic behavior of
$1-\alpha_k$? We have conjectured that, asymptotically,
$1-\alpha_k$ approaches $c_{\alpha}/k$, for some constant
$c_{\alpha} \geq 1$; but whether the conjecture holds, and if so for
which value of $c_{\alpha}$, remains to be seen.

Of mathematical interest, is the question whether $\alpha_k$ is a rational 
or an algebraic number. The answer could depend upon $k$. The recent work of
\cite{Tis22} on the algebraic nature of $\gamma_2$ indicates that this line of 
investigation may lead to uncovering deep combinatorial properties
that can shed light on various aspects of the subject, well beyond mere mathematical  curiosity.

The analysis of the statistical error could be somewhat improved, if
the experimental evidence that the standard deviation of $d_E(x,y)$ is
$o(\sqrt{n})$ were corroborated analytically. The dependence upon $k$
could also play a role here.

The complexity of computing, say, the most significant $h$ bits of
$\alpha_k(n)$ remains a wide open question. The only known lower bound
is $\Omega(\log k + \log n +h)$, based on the input and output size.
It is a far cry from the current upper bound, which is doubly
exponential in $\log n$. However, a deeper understanding of the distribution
and symmetries of the edit distance is likely to be required before
the computation time of $\alpha_k(n)$ can be significantly improved.

The lower bounds presented in Section~\ref{sec:lowerbound} are based
on upper bounds to the number of optimal scripts of a given cost $r$.
Our counting argument could be refined to take into account some
properties of optimal edit scripts. For example, in an optimal script,
an insertion cannot immediately precede or follow a deletion (since
the same result could be achieved by just one
substitution). Furthermore, it would be easy to show that, for many
pairs of strings, multiple scripts are counted by our argument. Part of the
difficulty with improving script counting comes from the
analytical tractability of the resulting combinatorial expressions,
which would be far more complicated than
Equation~(\ref{eq:script:bound}).

Another potential weakness of our lower bound is that it is derived
from a lower bound to $\ecc(x)$ that must hold for every $x\in
\Sigma_k^n$.  If, as $n$ goes to infinity, the fraction of strings
with eccentricity significantly smaller than the average were to
remain sufficiently high, then the approach would be inherently
incapable of yielding tight bounds.  On the other hand, preliminary
efforts seem to indicate that characterizing the strings with minimum
eccentricity is not straightforward. In contrast, it is a relatively
simple exercise to prove that the strings of maximum eccentricity are
those where all positions contain the same symbol and that their
eccentricity equals $(1-\frac{1}{k})n$.

In terms of applications, it would be interesting to explore the role
of statistical properties of the edit distance in string alignment and
other key problems in DNA processing and molecular biology.  One motivation is provided by
the error profile of reads coming from third generation sequencers
(\emph{e.g.}, PacBio), where sequencing errors can be modeled as edit
operations. In this context, it would be important to generalize the
analysis to non-uniform string distributions, whether defined
analytically or from empirical data, such as the distribution of
substrings from the human DNA. As briefly discussed in Sections~\ref{sec:montecarlo} and~\ref{sec:upperbound}, the approach we have presented for Monte Carlo estimates and for exact estimates via the Coalesced Dynamic Programming algorithm can easily handle arbitrary symbol distributions, as long as the symbols are statistically independent. The extension of the lower bounds of Section~\ref{sec:lowerbound} appears less straightforward.
Also of great interest would be to analyze the expected
distance between noisy copies of two independent strings, as well as
between a string and a noisy copy of itself, when the noise can be
modeled in terms of edit operations.




\paragraph{Author contribution}
Both authors contributed to the writing of and reviewed all sections of the manuscript.

\paragraph{Funding}
This work was supported in part by the European Union through MUR ``National Center for HPC, Big Data, and Quantum Computing'', CN00000013 (PNRR);
by MUR, the Italian Ministry of University and Research, under PRIN Project n. 2022TS4Y3N - EXPAND: scalable algorithms for EXPloratory Analyses of heterogeneous and dynamic Networked Data;
by MIUR, the Italian Ministry of Education, University and Research, under PRIN Project n. 20174LF3T8 - AHeAD: Efficient Algorithms for HArnessing Networked
Data;
by the University of Padova, Project CPGA$^3$
Parallel and Hierarchical Computing: Architectures, Algorithms,
and Applications;
and by IBM under a Shared University Research Award.

\bibliographystyle{alpha}
\bibliography{biblio}
\end{document}